%% file: main.tex
\documentclass[runningheads]{llncs}
\usepackage[T1]{fontenc}
\usepackage{fontawesome}
\usepackage{float}

\include{macros}

\usepackage{enumitem}
\usepackage[most]{tcolorbox}
\usepackage{cleveref}

\colorlet{fragmenttitle}{red!70!black}

\tcbset{frame hidden,boxrule=1pt,coltitle=fragmenttitle,colback=yellow!10!white, colframe=black, 
	detach title,before upper={\tcbtitle\quad},opacitybacktitle=0.25,
}

\tikzset{every picture/.append style={scale=0.8}}
\tikzset{every node/.style={scale=0.8}}

\title{Graphical Quadratic Algebra}
\titlerunning{Graphical Quadratic Algebra}
\author{
Dario Stein \inst{1} \and
Fabio Zanasi \inst{2}
\and
Robin Piedeleu \inst{2} \and
Richard Samuelson
\inst{3}
}
\authorrunning{Stein, Zanasi, Piedeleu, and Samuelson}
\institute{Radboud University Nijmegen, The Netherlands \email{dario.stein@ru.nl} \and University College London, UK \email{f.zanasi@ucl.ac.uk} \email{r.piedeleu@ucl.ac.uk} \and University of Florida, Gainesville, United States \email{rsamuelson@ufl.edu}
}

\begin{document}
\maketitle

\begin{abstract}
	Convex analysis and Gaussian probability are tightly connected,  as mostly evident in the theory of linear regression. Our work introduces an algebraic perspective on such relationship, in the form of a diagrammatic calculus of string diagrams, called Graphical Quadratic Algebra (GQA). We show that GQA is a complete axiomatisation for the category of quadratic relations, a compositional formulation of quadratic problems. Moreover, we identify a sub-theory of GQA which is complete for the category of Gaussian probabilistic processes. We show how GQA may be used to study linear regression and probabilistic programming.
\end{abstract}

\section{Introduction}\label{sec:intro}

Gaussian probability theory studies phenomena governed by \emph{normal} distributions --- bell-shaped curves fully determined by mean and variance. Quadratic optimisation, on the other hand, is concerned with minimising quadratic functions --- those expressible as polynomials of degree at most 2. Though they arise in different contexts, these two areas are deeply related. A clear instance of this connection appears in linear regression, where one seeks the best-fitting solution to a linear system $Ax = b$, with $A$ a matrix and $b$ a vector representing observed data. If the error term $\epsilon$ is assumed to follow a Gaussian distribution, then finding the most likely solution to the model $Ax + \epsilon = b$ is equivalent to minimising the quadratic loss  function $f(x) = ||Ax-b||^2$, which measures the sum of squared residuals. The relationship goes further: the log-density of a Gaussian distribution is itself a quadratic function, and computing conditional distributions in Gaussian models often reduces to solving least-squares problems.

In this paper, we introduce a fresh perspective on Gaussian probability, quadratic optimisation, and their web of connections, by presenting them as categorical structures. Our main results are \emph{complete axiomatisations} for these categories, via an axiomatic calculus called \emph{graphical quadratic algebra}.

Our approach unfolds as follows. First, we study how quadratic problems \emph{compose}, so that they can be organised into a category. To this aim, we view a partial quadratic function $F \colon \mathbb{R}^m \times \mathbb{R}^n \to [0,\infty]$ as a $[0,\infty]$-weighted relation on $\mathbb{R}^m \times \mathbb{R}^n$, called  \emph{quadratic relation}. This perspective has three major appeals. First, quadratic relations characterise the solutions of constrained least-square problems~(\cite{Boydbook}). Second, their composition is naturally defined via constrained minimisation (see~\eqref{eq:weighted_rel_comp} below). Third, quadratic relations extend affine relations (\cite{pbsz}), linking our work with the theory of \emph{Graphical Affine Algebra} (\cite{pbsz}) and offering a clear pathway towards a complete axiomatic calculus.

The resulting category $\quadrel$, with morphisms the quadratic relations, provides an algebraic setting to study quadratic optimisation. Gaussian probability enters the picture as a subcategory $\gauss$ of $\quadrel$. The category $\gauss$ has been previously studied as a \emph{Markov category}, in the context of a more general framework for categorical probability theory~\cite{fritz-markov}. Morphisms of $\gauss$, called Gaussian stochastic maps, generalise Gaussian distributions in the same way as linear maps generalise vectors. The study of $\gauss$ is of independent interest because of its use as a semantics for Gaussian probabilistic programming (\cite{stein2021compositional}). 

Once the `semantic' structures $\quadrel$ and $\gauss$ are in place, we introduce their `syntactic' presentation. Our language adopts the graphical notation of \emph{string diagrams} (see eg.~\cite{selinger:graphical,piedeleuzanasi2023} for recent introductions). As a two-dimensional formalism, string diagrams can be flexibly treated both as syntax and as combinatorial objects. They are now widely adopted in axiomatic reasoning on resource-sensitive models, eg. in linear algebra (\cite{bonchi:interacting_hopf,coecke2018,baez:control,CoeckeK18}) and probabilistic reasoning (\cite{fritz-markov,JacobsKZ21,LorenzTull-causalmodels,JacobsZanasi-logicalessentials}), making it a natural choice for our approach.

The language we introduce, called $\GQA$ (`Graphical Quadratic Algebra'), consists of the string diagrams obtained by sequential and parallel composition of the following generators, for $k \in \mathbb{R}$:

\vspace{-.5cm}
	\begin{align}
	\boxed{
	\input{bcomult.tikz} \quad \input{bcounit.tikz} \quad \input{wmult.tikz} \quad \input{wunit.tikz}  \quad \scalar{k} \quad \input{foot.tikz} 
	}\qquad
	\boxed{
	\input{bmult.tikz} \quad \input{bunit.tikz} \quad \input{wcomult.tikz} \quad \input{wcounit.tikz}
	}\qquad 
	\boxed{\NN}
	\label{eq:eqintro}
	\end{align}
\vspace{-.5cm}

As the name suggests, $\GQA$ directly extends \emph{Graphical Affine Algebra} ($\GAA$, \cite{pbsz}), which is given by the generators of the first two blocks in~\eqref{eq:eqintro}. Semantically, the generators in the first block model basic linear algebraic operations of copying ($\input{bcomult.tikz}$), discarding ($\input{bcounit.tikz}$) addition ($\input{wmult.tikz}$), scaling ($\input{scalar.tikz}$), and the constants zero ($ \input{wunit.tikz}$) and one ($\input{foot.tikz}$). With the equational theory of Hopf bimonoids (Fig.~\ref{fig:gaa-axioms}, first block), string diagrams built with these generators are known to axiomatise the category $\aff$ of affine maps (\cite{pbsz}). If we add the generators of the second block of~\eqref{eq:eqintro}, namely $\input{bmult.tikz},\input{bunit.tikz}$, $\input{wcomult.tikz},\input{wcounit.tikz}$, with behaviour symmetric to their mirrored version and extra equations (third block in Fig.~\ref{fig:gaa-axioms}, notably including two Frobenius monoid structures), we obtain the full $\GAA$, which axiomatises the category $\affrel$ of affine relations, i.e. affine subspaces that compose relationally~(\cite{pbsz}).

The surprising discovery of our work is that, in order to extend the axiomatisation of the category $\affrel$ of affine relations to one of the category $\quadrel$ of \emph{quadratic} relations, it suffices to add just one extra generator, $\NN$ (and three equations as in Fig.~\ref{fig:gaa-axioms}, second block): this yields $\GQA$. Moreover, we axiomatise $\gauss$ by restricting $\GQA$ to generators in the first and third block in~\eqref{eq:eqintro}.

\begin{figure}[t!]
	\begin{tcolorbox}[title={\small Affine fragment}]
		\begin{align*}
			\input{ax/add-associative.tikz}\;\myeq{as}\;\: \input{ax/add-associative-1.tikz}\qquad  \input{ax/add-commutative.tikz}\myeq{com}\;\;\, \input{wmult.tikz}\qquad  \input{ax/add-unital-left.tikz}\myeq{un}\;\input{id.tikz}
			\qquad
			\input{ax/foot-bcomult.tikz}\;\myeq{1-dup}\;\input{ax/footxfoot.tikz}
			\\
			\input{ax/copy-associative.tikz}\;\;\myeq{coas}\;\; \input{ax/copy-associative-1.tikz}\qquad  \input{ax/copy-commutative.tikz}\;\;\myeq{cocom}\;\; \Bcomult\qquad  \input{ax/copy-unital-left.tikz}\myeq{coun}\;\input{id.tikz}
			\qquad 
			\input{ax/foot-bcounit.tikz}\;\myeq{1-del}\;\;\;\input{empty-diag.tikz}
			\\
			\input{ax/reals-add.tikz}\;\myeq{add}\;\input{ax/reals-add-1.tikz} \qquad \input{zero.tikz}\;\myeq{zero}\;\input{ax/reals-zero.tikz} \qquad
			\input{ax/reals-copy.tikz}\;\myeq{dup}\; \input{ax/reals-copy-1.tikz} \qquad \input{ax/reals-delete.tikz}\;\myeq{del}\;\Bcounit
			\quad
			\\
			\input{ax/reals-multiplication.tikz}\;\myeq{mult}\;\input{ax/reals-multiplication-1.tikz} \qquad    \input{ax/reals-sum.tikz}\;\myeq{plus}\;\input{ax/reals-sum-1.tikz}
			\qquad    \input{ax/reals-scalar-zero.tikz}\;\myeq{0}\;\input{ax/reals-scalar-zero-1.tikz}\qquad  \input{ax/reals-scalar-one.tikz}\myeq{1}\;\input{id.tikz}
			\\
			\input{ax/add-copy-bimonoid.tikz}\;\;\myeq{$\circ\bullet$-bi}\;\;\input{ax/add-copy-bimonoid-1.tikz} \qquad \input{ax/add-copy-bimonoid-unit.tikz} \,\;\;\myeq{$\circ\bullet$-biun}\;\;\; \input{ax/add-bimonoid-unit-1.tikz}
			\qquad
			\input{ax/bone-white-black.tikz}\;\;\myeq{$\circ\bullet$-bo}\;\;\;\input{empty-diag.tikz}
			\qquad
			\input{ax/add-copy-bimonoid-counit.tikz} \;\;\,\myeq{$\bullet\circ$-biun}\;\;\; \input{ax/add-copy-bimonoid-counit-1.tikz}
		\end{align*}
	\end{tcolorbox}
	\begin{tcolorbox}[title={\small Quadratic fragment}]
		\begin{align*}
			\input{eqD-left.tikz}\;\myeq{D}\, \input{empty-diag.tikz} \,\myeq{Z}\; \input{eqNorm0-left.tikz}  \qquad
			\input{eqRI-left.tikz}\,\myeq{RI}\;\input{eqRI-right.tikz} \quad \text{for $\varphi \in [0,2\pi)$}
		\end{align*}
	\end{tcolorbox}
	\begin{tcolorbox}[title={\small Relational fragment}]
		\begin{align*}
			\input{ax/copy-Frobenius-left.tikz}\;\;\myeq{$\bullet$-fr1}\;\; \input{ax/copy-Frobenius.tikz}\;\;\myeq{$\bullet$-fr2}\;\; \input{ax/copy-Frobenius-right.tikz}  \qquad
			\input{ax/copy-special.tikz}\myeq{$\bullet$-sp}\input{id.tikz}
			\qquad 
			\input{ax/bone-black.tikz}\;\;\myeq{$\bullet$-bo}\;\;\input{empty-diag.tikz} \qquad \quad
			\\
			\input{ax/add-Frobenius-left.tikz}\;\;\myeq{$\circ$-fr1}\;\; \input{ax/add-Frobenius.tikz}\;\;\myeq{$\circ$-fr2}\;\; \input{ax/add-Frobenius-right.tikz}
			\qquad
			\input{ax/add-special.tikz}\myeq{$\circ$-sp}\input{id.tikz}
			\qquad 
			\input{ax/bone-white.tikz}\;\;\myeq{$\circ$-bo}\;\;\input{empty-diag.tikz}\qquad\quad
			\\
			\input{ax/white-cap.tikz}\;\myeq{cap}\; \input{ax/black-cap-minus-one.tikz} 
			\quad\; 
			\input{one-false.tikz}\; \myeq{false}\;\input{one-false-disconnect.tikz}
			\qquad
			\input{ax/scalar-division.tikz}\;\myeq{$r$-inv}\; \input{id.tikz}\;\myeq{$r$-coinv}\;\input{ax/scalar-co-division.tikz} (r\neq 0)
		\end{align*}
	\end{tcolorbox}
	\caption{Axioms of Graphical Quadratic Algebra \label{fig:gaa-axioms}}
\end{figure}

 In $\gauss$, $\NN$ is interpreted as the standard normal distribution $\N(0,1)$, and in $\quadrel$ as the quadratic function $f(x) = \frac 1 2 x^2$ (ie., a unary quadratic relation on $\mathbb{R}$). We illustrate this dual interpretation via an example.

\begin{example}\label{ex:intro}
	The semantics of the string diagram $\input{add_example.tikz}$ in $\quadrel$, built up from two instances of $\NN$ and $\input{wmult.tikz}$, is the constrained minimisation problem on the left below, whose solution can be found on the right-hand side, with its corresponding diagram:
	\begin{equation}\label{eq:intro_den_eq} 
		\hspace{-.5cm}\sem{\scalebox{0.7}{\input{add_example.tikz}}\!\!} \!=\! \left( y \mapsto \inff {} {\frac 1 2 x_1^2 + \frac 1 2 x_2^2 : x_1 + x_2 = y }\right) \!=\! \left(y\mapsto  \frac 1 4 y^2 \right) \!=\! \sem{\input{add_example_2Right.tikz}}
	\end{equation}
	In $\gauss$, \eqref{eq:intro_den_eq} means that if $X_1, X_2 \sim \N(0,1)$ are independent variables, then $Y=X_1+X_2$ has distribution $\N(0,2)$.
	
	We can now use the axioms of GQA (Fig.~\ref{fig:gaa-axioms}) to derive \eqref{eq:intro_den_eq}. The notable step involves the axiom scheme \textsf{RI}, expressing that the function $f(x,y) = \frac 1 2 (x^2 + y^2)$ is \emph{invariant under rotations}. When interpreted in $\gauss$, \textsf{RI} is  evocative of the Herschel-Maxwell theorem, which states that the normal distribution is uniquely characterised by its rotation invariance (e.g. \cite{gyenis2017maxwell}).
 Because $\cos(\pi/4) = \sin(\pi /2 - \pi /4) = \sin(\pi /4) = 1/\sqrt 2$, we can apply \textsf{RI} with $\varphi = \pi/4$ as below. We then simplify the resulting string diagram via linear algebraic manipulations in the fragment GAA, obtaining the desired outcome $\input{add_example_2Right.tikz}$.
	\[ \input{add_example_3unpacked_Fabio.tikz} \]
\end{example}

Completeness of $\GQA$ for $\quadrel$ and $\gauss$ ensures that any semantic equality as in Example~\ref{ex:intro} can in principle be derived. More broadly, this algebraic characterisation sheds light on the connections between quadratic optimisation and Gaussian probability theory, revealing their relative expressive power. Notably, the equational theory of $\gauss$ is a fragment of the one of $\quadrel$, without the ability to model relational behaviours using Frobenius monoids. Moreover, the presentation of $\gauss$ omits the generator $\input{wcounit.tikz}$, which we may regard as conditioning of Gaussians. From this perspective, we may view quadratic relations as generalised Gaussian distributions, which combine probability and conditioning.

There are a few more interesting ramifications of our work. As noted, there are transformations relating Gaussian distributions, quadratic problems, and affine spaces. However, turning them into functors is technically challenging, as studied in~\cite{stein2023compositional}. Thanks to our axiomatic presentations, it suffices to define these mappings on the generators of $\GQA$ to obtain functoriality as a simple corollary. We explore this in~\Cref{sec:functors}. 

Our characterisations also enable a new methodology, based on equational diagrammatic reasoning, to study domains modelled by quadratic relations or Gaussian distributions, such as linear regression and probabilistic program semantics. We explore these examples in Section~\ref{sec:cases-studies}. As a preview, the program below left is written in a language with support for Gaussian distributions and a conditioning operator \mlstinline{(=:=)}, and computes inference of a latent variable from a noisy observation. The corresponding $\GQA$-diagram is below right. 

\noindent 
\hspace{-.3em} 
\fbox{
\begin{minipage}{.3\textwidth}
	\vspace{-.4cm}
	\begin{align*}
	&\letin {\text{value}} {50 + 10 \cdot \normal()} 
			\\
			&\letin {\text{measurement}} {\text{value} + 5 \cdot \normal()} 
			\\
			&\text{measurement} \eq \underline{40}; 
			\\
			&\text{value}
	\end{align*}
\end{minipage}}
\begin{minipage}{.4\textwidth}
	\vspace{-.4cm}
	\begin{equation*} 
		\input{noisymmt.tikz} 
	\end{equation*}
\end{minipage}

 Programs in this language correspond systematically to string diagrams. Moreover, the axioms of $\GQA$ suffice to reduce such a diagram to one representing the posterior distribution, effectively mirroring the program's execution.

\noindent \textbf{Synopsis} After preliminaries (\Cref{sec:GAA}), we present our work  in two steps: first we introduce the simpler theory $\gauss$, and show how the fragment $\GQAfwd$ axiomatises it (\Cref{sec:gauss}). Next, we introduce $\quadrel$ and its axiomatic calculus, the full $\GQA$ (\Cref{sec:quadrel}). We then conclude with some applications~(\Cref{sec:cases-studies}) and future work (\Cref{sec:conclusions}). Missing proofs and details are in appendices.


\section{Background: Graphical Affine Algebra}\label{sec:GAA}

We recall background on string diagrams and Graphical Affine Algebra,  referring to~\cite{pbsz} for more details. A symmetric strict monoidal category with set of objects the natural numbers, and monoidal product on objects given by addition, is called a \emph{prop} (\cite{Maclane1960}). Prop morphisms are identity-on-objects symmetric monoidal functors between props. Props play a role akin to algebraic clones, in a monoidal setting. Various linear algebraic transformations may be organised into props. 

\begin{definition}\label{def:cataffine}~The prop $\aff$ has affine maps of type $\mathbb{R}^m \to \mathbb{R}^n$ as morphisms $m\to n$, with composition given by matrix multiplication and monoidal product given by direct sum.
	The prop $\affrel$ has affine subspaces of $\mathbb{R}^{m+n}$ as morphisms $m \to n$, with relational composition $R ; S = \{ (\vecv,\vecw) \mid \exists \vecu . (\vecv,\vecu ) \in R \wedge (\vecu,\vecw ) \in S \}$, and monoidal product given by direct sum.
\end{definition}
Because of the way they compose, $\affrel$-morphisms are called \emph{affine relations}. Complete axiomatisations for these categories are provided by string diagrammatic calculi, which we now introduce. A \emph{symmetric monoidal theory} (SMT) is a pair $(\Sigma, E)$, where $\Sigma$ is a signature of operations $o \colon m \to n$ with an arity $m$ and coarity $n$, and $E$ is a set of equations between $\Sigma$-terms. A $\Sigma$-term $c$ of type $m \to n$ will be represented graphically as \emph{a string diagram} (\cite{piedeleuzanasi2023,selinger:graphical}) with $m$ dangling wires on the left and $n$ on the right, also written $\input{cdiagrammn.tikz}$. Formally, $\Sigma$-terms are freely obtained by  sequential and parallel compositions of the operations in $\Sigma$ together with the identity $\input{id.tikz} \colon 1 \to 1$, the symmetry $\input{swap.tikz} \colon 2 \to 2$, and the `empty' diagram $\input{empty-diag.tikz} \colon 0 \to 0$. Sequential composition of $\Sigma$-terms $\input{cdiagram.tikz}$ and $\input{ddiagram.tikz}$ is depicted as $\input{horizontal-comp.tikz}$, of type $m \to n$. Parallel composition of $\Sigma$-terms $\input{c1diagram.tikz}$ and $\input{c2diagram.tikz}$ is depicted as $\input{vertical-comp.tikz}$, of type $m_1 + m_2 \to n_1 + n_2$. We refer to $\Sigma$-terms quotiented by the laws of symmetric strict monoidal categories (Appendix~\ref{app:lawssmc}) as \emph{string diagrams}.\footnote{We adopt a `syntactic' definition of string diagrams, which follows eg. \cite{piedeleuzanasi2023,pbsz,bonchi:interacting_hopf} and is more convenient for our axiomatic perspective. The reader may be familiar with the combinatorial (\cite{Bonchi_Gadducci_Kissinger_Sobocinski_Zanasi_2022}) and topological (\cite{JoyalStreet77}) view on string diagrams: the three are ultimately equivalent. Additional details may be found eg. in~\cite{piedeleuzanasi2023}.} 

Given an SMT $(\Sigma, E)$, the prop $\freeprop{(\Sigma, E)}$ freely generated by $(\Sigma, E)$ has morphisms $m \to n$ the string diagrams of type $m \to n$ quotiented by $E$, with sequential and parallel composition defined as on the corresponding $\Sigma$-terms. We will typically regard $\freeprop{(\Sigma, E)}$ as a \emph{string diagrammatic calculus}, with the string diagrams being its (graphical) \emph{syntax}. The props defined `directly', as $\aff$ and $\affrel$, are regarded as \emph{semantic models} for diagrammatic calculi. When there is an isomorphism of props $\sem{\cdot}$ between $\freeprop{(\Sigma, E)}$ and another prop $\catname{C}$, we say that $\catname{C}$ is \emph{presented} (or axiomatised) by $(\Sigma,E)$. Because prop morphisms are identity-on-objects, to prove such result it suffices to prove that $\sem{\cdot}$ is full and faithful. In logical terms, we can phrase these requirements as \emph{soundness} (if $s = t$ in $\freeprop{(\Sigma, E)}$, then $\sem{s} = \sem{t}$), \emph{completeness} (if $\sem{s} = \sem{t}$ in $\catname C$, then $s = t$ in $\freeprop{(\Sigma, E)}$) and \emph{definability} (for every $f \colon m \to n$ in $\catname C$, there exists a $\Sigma$-term $s : m \to n$ with $\sem{s} = f$).

\emph{Graphical Affine Algebra} is based on the following operations:

\vspace{-.3cm}
\noindent\begin{minipage}{.5\textwidth}
\begin{equation}
	\input{bcomult.tikz} \quad \input{bcounit.tikz} \quad \input{wmult.tikz} \quad \input{wunit.tikz}  \quad \scalar{k} \quad \input{foot.tikz} \label{eq:genGAAfwd}
	\end{equation}
\end{minipage}
\begin{minipage}{.5\textwidth}
\begin{equation}
	\quad\qquad\input{bmult.tikz} \quad \input{bunit.tikz} \quad \input{wcomult.tikz} \quad \input{wcounit.tikz} \qquad \label{eq:genGAA}
\end{equation}
\end{minipage}
\vspace{.2cm}

Row~\eqref{eq:genGAAfwd} denotes basic linear transformations such as copying ($\input{bcomult.tikz}$), discarding ($\input{bcounit.tikz}$) addition ($\input{wmult.tikz}$), the constant zero ($\input{wunit.tikz}$), multiplication by scalar $k$ ($\scalar{k}$), and the constant one ($\input{foot.tikz}$). Row~\eqref{eq:genGAA} expresses the converse of these maps, regarded as relations (subspaces), with the `mirror' of $\input{scalar.tikz}$ and $\input{foot.tikz}$ derivable as follows: $\input{coscalar.tikz}  :=  \input{coscalar-def.tikz}$ and $\input{cofoot.tikz}  :=  \input{cofoot-def.tikz}$.

The affine fragment in Fig.~\ref{fig:gaa-axioms} describes the interaction between operations in~\eqref{eq:genGAAfwd}: $\input{wmult.tikz}$, $\input{wunit.tikz}$ form a commutative monoid, $\input{bcomult.tikz}$, $\input{bcounit.tikz}$ form a commutative comonoid, distributing over each other according to the laws of bimonoids. Furthermore, scalars and $\input{foot.tikz}$ distribute over the bimonoid, and there are laws expressing the field structure on $\mathbb{R}$. 

The relational fragment in Fig.~\ref{fig:gaa-axioms} describes additional laws involving operations in~\eqref{eq:genGAA} and $\input{coscalar.tikz}$, $\input{cofoot.tikz}$ defined above: black and white special Frobenius monoids, the equivalence between the black and the white compact closed structure arising from the Frobenius algebras (\textsf{cap}), the fact that scalars have inverses (\textsf{$r$-inv}, \textsf{$r$-coinv}), and an ``ex-falso quodlibet'' principle (\textsf{false}, where $\input{false.tikz}$ expressed the `impossible' condition that $1 = 0$). Note at this stage we may also prove that $\input{wcomult.tikz}$, $\input{wcounit.tikz}$ form a commutative comonoid and $\input{bmult.tikz}$, $\input{bunit.tikz}$ form a commutative monoid. We write $\GAAfwd$ for the prop freely generated by operations~\eqref{eq:genGAAfwd} and equations in the affine fragment. Similarly, we write $\GAA$ for the prop freely generated by operations~\eqref{eq:genGAAfwd}-\eqref{eq:genGAA}, and equations in the causal and relational fragments. These props present affine maps and affine relations respectively. 
\begin{proposition}[\cite{pbsz}]\label{prop:presentationsGLA} $\GAAfwd$ presents $\aff$ and  $\GAA$ presents $\affrel$.
\end{proposition}
Importantly, Proposition~\ref{prop:presentationsGLA} implies that, whenever we reason in the aforementioned calculi, any equivalence of matrices or subspaces represented by string diagrams are provable in the corresponding equational theory. This will allow us to use freely known facts of linear algebra (e.g. those recalled in Appendix~\ref{app:linearalgebra}) in diagrammatic reasoning whenever needed, allowing us to focus on the quadratic extension that is the topic of this paper. 

In preparation for our completeness theorems, we recall standard notation to represent matrices and subspaces as string diagrams of $\GAA$. First, for readability, we use thick wires for an arbitrary number of ingoing/outgoing wires in a string diagrams, as in $\input{cdiagram-thick.tikz} := \input{cdiagram.tikz}$. We use a similar convention to represent multiple instances of the same generator, e.g. $\input{bcomult-thick.tikz} := \input{bcomult-thick-def.tikz}$, $\input{bcomult-thick.tikz} := \input{bcomult-thick-def.tikz}$, and $\input{wunit-thick.tikz} := \input{wunit-thick-def.tikz}$. An $n \times m$ matrix $A$ may be represented by a string diagram $\input{matrixdiagram.tikz} \colon m \to n$, as follows: the wires on the left of $\input{matrixdiagram.tikz}$ stand for the columns of $A$, the wires on the right stand for the rows, and the left $j$-th wire is connected to the $i$-th wire on the right through a scalar $\scalar{k}$ if the coefficient $A_{ij}$ is $k$ (see~\cite{DBLP:phd/hal/Zanasi15} for details). In particular, up to the axioms of $\GAA$, two wires are disconnected when $A_{ij} = 0$, and are just a plain wire when $A_{ij} =1$. For instance,
\begin{equation}\label{eq:matrixform}
	\text{if }\ A =  { \begin{pmatrix}
			\small k_1 & \small 0 & \small 0 \\
			\small k_2 & \small 0 & \small 1 \\
			\small 1 & \small 0 & \small 0 \\
			\small 0 & \small 0 & \small 0
	\end{pmatrix}} \ \text{then }\ \input{matrixdiagram.tikz} := \input{ex-matrix.tikz} \ \ .\end{equation}
Here , $\input{matrixdiagram.tikz}$ represents $A$ in the sense that it is interpreted as $A \colon n \to m$ via the isomorphism $\GAAfwd \to \aff$, and as $\{ (\vecv,\vecw) \mid \vecw = A \vecv\}$ via the isomorphism $\GAA \to \affrel$. Similarly, a subspace $S \subseteq \mathbb{R}^{m+n}$ may be represented by the string diagram $\input{subspacediagram.tikz}$ in $\GAA$, for any $A$ such that $\im(A) = S$. We will also write this as $\input{subspacediagram_compressed.tikz}$. Well-definedness of this encoding is justified by \Cref{prop:subspace_encoding_uniqueness}, App.~\ref{app:linearalgebra}. Finally, given a scalar $c \in \R$, we sometimes abbreviate $\input{scalar-factor.tikz}$ with $\input{scalar-factor-notation.tikz}$.

\section{Axiomatising Gaussian Probability}\label{sec:gauss}

\subsection{Gaussian Probability}\label{sec:gaussprob}

Gaussian probability is an expressive self-contained fragment of probability theory. We recall what is strictly necessary here, referring to e.g.~\cite{bill1995prob} for further background. A random variable $X$ is called \emph{normally distributed} or \emph{Gaussian} with mean $\mu \in \mathbb{R}^n$ and variance $\sigma^2$ for $\sigma \in [0, \infty)$, if it has density function
$f(x) = \frac 1 {\sqrt{2\pi\sigma^2}} \exp \left(-\frac{(x-\mu)^2}{2\sigma^2}\right)$
with respect to the Lebesgue measure. This is typically written $X \sim \N(\mu,\sigma^2)$, meaning the distribution $\N(\mu,\sigma^2)$, also called Gaussian distribution, is the law of $X$. The \emph{standard normal distribution} $\N(0,1)$ is particularly significant, as  it generates any other normal distribution, in the following sense:
a random vector $\vec X$ is \emph{multivariate normal} if its distribution arises as $\vec X = A \cdot (Z_1, \ldots, Z_k) + \vec \mu$,
where $Z_1, \ldots, Z_k \sim \N(0,1)$ are independent variables and $A \in \R^{n \times k}$ is a matrix and $\vec \mu \in \R^n$. In other words, multivariable normal distributions arise as the pushforwards of standard normal distributions under affine maps. The distribution of $\vec X$ is fully characterised by its mean $\vec \mu$ and its covariance matrix $\Sigma = AA^T \in \mathbb{R}^{n \times n}$, for $A \in \mathbb{R}^{n \times k}$. Conversely, every positive semidefinite matrix $\Sigma$ is the covariance matrix of a unique Gaussian distribution $\N(\mu,\Sigma)$, supported on the affine subspace $\mu + \im(\Sigma)$.
A convenient way of building up complicated Gaussian distributions is by pushing forward existing distributions under affine maps: if $\N(\mu,\Sigma)$ is a Gaussian distribution and $A$ is a matrix, then the pushforward distribution $A_*\N(\mu,\Sigma)$ is $\N(A\mu,A\Sigma A^T)$. 
 If we aim to include Gaussian distributions and affine maps in the same categorical structure, it is natural to study \emph{Gaussian maps}, which are a simple kind of stochastic map consisting of an affine function with Gaussian noise, informally written $f(x) = Ax + \N(b,\Sigma)$, with $A \in \R^{n \times m}$, $b \in \R^n$ and $\Sigma \in \R^{n \times n}$.
To compose this with another Gaussian map $g(y) = Cy + \N(d,\Xi)$, we combine their Gaussian noise independently using the transformation rule $g(f(x)) = CAx + \N(Cb + d, C\Sigma C^T + \Xi)$.

\begin{definition}[{\cite[\S~6]{fritz-markov}}] \label{def:gauss}
	The prop $\gauss$ has as morphisms $m \to n$ the tuples $(A,b,\Sigma)$ with $A \in \R^{n \times m}$, $b \in \R^n$ and positive semidefinite $\Sigma \in \R^{n \times n}$. Composition and monoidal product are given by
	\begin{align*}
		(A,b,\Sigma) ; (C,d,\Xi) &= (CA,Cb+d,C\Sigma C^T + \Xi) \\
		(A_1,b_1,\Sigma_1) \oplus (A_2,b_2,\Sigma_2) &= \begin{pmatrix}
			\begin{pmatrix} \scriptstyle A_1 \!\!& \!\! \scriptstyle 0 \\ \scriptstyle 0 & \scriptstyle A_2 \end{pmatrix},
			\begin{pmatrix} \scriptstyle b_1 \\ \scriptstyle b_2 \end{pmatrix},
			\begin{pmatrix}\scriptstyle  \Sigma_1 \!\!&\!\! \scriptstyle \scriptstyle 0 \\ \scriptstyle 0 \!&\! \scriptstyle \Sigma_2 \end{pmatrix}
		\end{pmatrix}
	\end{align*}
\end{definition}

\noindent Gaussian distributions are the `states' $0 \to n$ in $\gauss$. There is a prop morphism $\vect \to \gauss$ representing an affine map $x \mapsto Ax + b$ as the Gaussian map $(A,b,0)$.

\subsection{Causal Graphical Quadratic Algebra and Completeness}
\label{sec:causal-GQA}

We now turn our attention to the axiomatic theory of $\gauss$. The prop $\GQAfwd$, which stands for \emph{causal} Graphical Quadratic Algebra, is freely generated by the following operations
\begin{equation}\label{eq:genGQAfwd}
	\begin{aligned}
		\input{bcomult.tikz} \quad \input{bcounit.tikz} \qquad \input{wmult.tikz} \qquad \input{wunit.tikz}  \qquad \scalar{k} \qquad \input{foot.tikz} \qquad \NN 
	\end{aligned}
\end{equation}
where $k$ ranges over $\mathbb{R}$, and equations as in the first two blocks of Fig.~\ref{fig:gaa-axioms} (linear and quadratic fragments), \emph{except} for axiom $\textsf{Z}$ (where the generator $\input{wcounit.tikz}$ appears, which is not among those in~\eqref{eq:genGQAfwd}). As in related work~\cite{fritz-markov,cho_jacobs}, the term `causal' indicates that the operations in~\eqref{eq:genGQAfwd} should be read as processing values from left-to-right, without constraining the values of inputs. Note we may regard $\GQAfwd$ as an extension of $\GAAfwd$ from Section~\ref{sec:GAA}. The extension is simply given by adding the operation~$\NN$ and the associated equations. It is also a fragment of full $\GQA$, which we will introduce later.
We now define an interpretation $\semG{\cdot}$ of $\GQAfwd$-generators as morphisms of $\gauss$. We use the notation $f(x) = Ax + \N(b,\Sigma)$ for such a morphism of type $m \to n$.
\begin{equation*}
	\begin{gathered}
		\semG{\input{bcomult.tikz}}(x) \ =\  {x \choose x} + \N\left({0 \choose 0},\begin{pmatrix}
			0 & 0 \\
			0 & 0
		\end{pmatrix}\right) \qquad
		\semG{\input{bcounit.tikz}}(x)  \ =\  \nullm + \N(\nullm,\zeromat) 
		\\ 
		\semG{\input{wmult.tikz}}{x_1 \choose x_2} \ =\  \left( x_1 + x_2 \right) + \N(0,0)
		\qquad
		\semG{\input{wunit.tikz}}\nullm \ =\  (0) + \N(0,0)
		\\ 
		\semG{\scalar{k}}(x)   = (k \cdot x) + \N(0,0) \quad
		\semG{\input{foot.tikz}}\nullm  =  (0) + \N(1,0)
		\quad
		\semG{\NN}\nullm =  (0) + \N(0,1) 
	\end{gathered}
\end{equation*}

Note that the newly added operation $\NN$ is the only one with a non-trivial probabilistic component. 
The assignment extends\footnote{This holds provided that $\semG{\cdot}$ is sound, namely $c =d$ in GQA implies $\semG{c}=\semG{d}$. This can be readily verified on the axioms of Fig.~\ref{fig:gaa-axioms}, see Appendix~\ref{app:soudness}. For a comprehensive explanation of the universal property of freely generated props, which is used both here and in Section~\ref{sec:functors}, see~\cite{baez:props}.} to a prop morphism $\semG{\cdot}\colon$ $\GQAfwd \to \gauss$, by interpreting sequential and parallel compositions of string diagrams as the corresponding operations in $\gauss$. In fact, we can show that the equational theory of Fig.~\ref{fig:gaa-axioms} makes $\semG{\cdot}$ an isomorphism.

\begin{theorem}\label{thm:completenessgauss}
	$\gauss$ is presented by the fragment $\GQAfwd$.
\end{theorem}
\begin{proof}
	For definability, let $f \in \gauss(m,n)$ be given by $f(x) = Ax + \N(\mu,\Sigma)$. Then, for $\Sigma = LL^T$, where $L$ is lower triangular, 
		\begin{equation} \label{eq:gauss_normal_form}
	\semG{\input{gauss_definability_expanded.tikz}} = f.
	\end{equation}
For completeness, we may transform any string diagram of $\GQAfwd$ into the form \eqref{eq:gauss_normal_form}, see \Cref{prop:nf_gauss} in Appendix~\ref{app:completeness-causal-gqa}. This is a normal form, because we can read off the values of $(A,\mu,\Sigma)$ unambiguously from it (\Cref{prop:encoding_uniqueness}). Note a key part in uniqueness is played by the property that $\NN$ is invariant under arbitrary orthogonal matrices, a generalisation of \textsf{RI} (\Cref{prop:orthogonalinv}). Finally, $\semG{(A,\mu,\Sigma)} = \semG{(A', \mu',\Sigma')}$ implies $(A,\mu,\Sigma) = (A', \mu',\Sigma')$, yielding completeness.
\end{proof}

\section{Axiomatising Quadratic Problems}\label{sec:quadrel}

As mentioned in the introduction, Gaussian probability and quadratic problems are tightly related. We will now introduce \emph{quadratic relations} as a categorical formalism to express quadratic problems in a compositional way. We then give a presentation $\GQA$ for this category, and use this presentation to study the relationship with Gaussian probability in an algebraic fashion.

\subsection{Quadratic Problems as Quadratic Relations}\label{sec:backgroundQuadraticRelations}

We justify the introduction of quadratic relations in steps. The basic setup of an optimisation problem is: given a weight (or cost) function $f(x_1,\ldots, \allowbreak x_{m+n})$, compute the parameterised infimum $(x_1,\ldots,\allowbreak x_m) \mapsto \inf \{f(x_1,\ldots, \allowbreak x_{m+n}) : $ $(x_{m+1},\dots, \allowbreak x_{m+n}) \in \R^n\}$ over some variables. We can type such a weight function as $F : \R^m \times \R^n \to [0,+\infty]$, thereby formally specifying the `input' and `output' variables. How to compose such functions? We can see $F : \R^m \times \R^n \to [0,+\infty]$ as a \emph{weighted relation} the same way that an ordinary relation can be seen as $R : \R^m \times \R^n \to \{0,1\}$. This suggests a notion of composition: while relations compose as on the left below,
weighted relations compose by adding the weights and minimising over the shared variable\footnote{This analogy can be made rigorous with quantales~(\cite{lawvere-metric}). Our weighted relations are parametrised over the quantale $([0,+\infty],\geq,+)$.}, as on the right below:
\begin{equation} 
	(R;S)(x,z) = \bigvee_{y} R(x,y) \cdot S(y,z)\qquad (F;G)(x,z) = \inff y {F(x,y) + G(y,z)} \label{eq:weighted_rel_comp}
 \end{equation}
This way of building up optimisation problems is well-established, see e.g.~\cite{hanks2024modeling} for model-predictive control, \cite{rockafellar} for convex analysis (where weighted relations are called \emph{bifunctions}),~\cite{willems:oss} for control theory, and \cite{stein2023compositional} for a categorical perspective. In our developments, it is convenient to introduce the following notation, reminiscent of Iverson brackets~(\cite{rockafellar}), associating a weighted relation $\iv{\phi}$ to a formula $\phi$, denoting a linear equation between its free variables:
\begin{equation} \label{eq:iversonbrackets}
	\iv{\phi} \defeq 
		0 \quad\text{if $\phi$ is true}, \quad \infty \quad \text{otherwise}
\end{equation}
Note that $0$ denotes truth and $\infty$ falsity in this context! The indicator function of a set $A$ is given by $1_A(x) = \iv{x \in A}$.
\begin{example}\label{ex:compositionGQA}
	Consider the weighted relation $F : \R \times \R^2 \to [0,\infty]$ given by the constraint
	$F(c,(x,y)) = \iv{c = x + y}$, and let $G : \R^2 \times \R^0 \to [0,\infty]$ be $G((x,y),()) = \frac 1 2 (x^2+y^2)$. Then the composite $F ; G$ computes the parameterised quadratic problem $(F;G)(c,()) = \inff{} {\frac 1 2 (x^2 + y^2) : x + y = c } = \frac 1 4 c^2$. We will render this composition diagrammatically in $\GQA$, \emph{cf.}~\Cref{ex:compositionGQAdiag}.
\end{example}

We now focus on the class of weighted relations, called \emph{quadratic relations}, which correspond to quadratic problems: this is a type of optimisation problem involving quadratic objective functions and affine constraints. First, a \emph{quadratic function} $f : \R^n \to \R$ is a function that can be written as a multivariate polynomial of degree at most $2$. Examples of quadratic functions are the squared euclidean distance $||\vec x||^2$, $(x_1-1)x_2$, and $x_1^2-x_2$. A \emph{partial quadratic function} is allowed to assume the value $+\infty$ (indicating falsity/partiality) outside an affine subspace $M$. 

\begin{definition}[{{\cite[p.109]{rockafellar}}}]
	A \emph{partial quadratic function} $f : \R^n \to \exR$ is a function that can be written in form
	$f(x) = \langle x, \Sigma x\rangle + \langle b,x \rangle + c + \iv{x \in M}$
	where $M \subseteq \R^n$ is an affine subspace, $\Sigma \in \R^{n \times n}$ is a symmetric matrix, $b \in \R^n$, $c \in \R$, and $\langle -,- \rangle$ denotes the standard inner product on $\R^n$.
\end{definition}

Recall that a function $f : \R^n \to \exR$ is \emph{convex} if its epigraph $\{ (x,y) : y \geq f(x) \}$ is a convex subset of $\R^{n+1}$. An \emph{elementary convex partial quadratic function} is one in diagonal form $h(x) = \sum_i \lambda_i x_i^2$ with $\lambda_i \in [0,+\infty]$. The function $h$ is partial, because $h(x) = \infty$ whenever $\lambda_i = \infty$ and $x_i \neq 0$.

\begin{example}\label{ex:gausslogdesnity}
The negative log-density function of a Gaussian distribution on $\R^n$ is a partial convex quadratic function on $\R^n$. It is partial because it takes the value $+\infty$ (corresponding to density $0$) outside of its support subspace $S$. 
\end{example}

It turns out that nonnegative partial quadratic functions are automatically convex. Also, they are closed under the composition formula \eqref{eq:weighted_rel_comp}. We state a slightly more general proposition, namely closure under arbitrary constrained infima (see Appendix~\ref{app:quadraticproofs}).
\begin{proposition}\label{lemma:quadrel_infimization}
	If $M \subseteq \R^n \times \R^m$ is an affine subspace, and $g$ is a nonnegative partial quadratic function, then so is $f(x) = \inff {} { g(y) \mid (x,y) \in M }$.
\end{proposition}
\begin{definition}[Quadratic Relations]\label{def:quadrel}
	The prop $\quadrel$ of quadratic relations has:
	\begin{itemize}[topsep=0pt,itemsep=-1ex,partopsep=1ex,parsep=1ex]
		\item morphisms $m \to n$ are non-negative partial quadratic functions $F : \R^m \times R^n \to [0,\infty]$, called \emph{quadratic relations}
		\item identities are indicator functions $\id_{n}(x,y) = \iv{x=y}$
		\item composition is by minimisation: for $F \colon m \to n$ and $G \colon n \to p$, define $F; G : m \to p$ as $(F;G)(x,z) := \inff y {F(x,y) + G(y,z)}$.
		\item the monoidal product is addition. Given $F_i \colon {m_i} \to {n_i}$ for $i=1,2$, then $F_1 \oplus F_2 \colon {m_1+m_2} \to {n_1+n_2}$ is defined by
		$ (F_1 \oplus F_2)((x_1,x_2),(y_1,y_2)) = F_1(x_1,y_1) + F_2(x_2,y_2). $
	\end{itemize}
\end{definition}

\noindent Composition in $\quadrel$ is well-defined by \Cref{lemma:quadrel_infimization}, since we can write  $(F;G)(x,z) =$ $\text{inf} \{H(\vec y) :$ $(x,z,\vec y) \in M \}$ where $H(y_1,y_2,y_3,y_4) = F(y_1,y_2) + G(y_3,y_4)$ and $M = \{ (x,z,y_1,y_2,y_3,y_3) : x = y_1, y_2 = y_3, z = y_4 \}$. 
To every affine relation $M \subseteq \R^m \times \R^n$, we can associate the quadratic relation $1_M$ given by its indicator function. This assignment respects composition, i.e. yields a prop morphism $\affrel \to \quadrel$.

\subsection{Graphical Quadratic Algebra}

We now introduce the diagrammatic calculus presenting $\quadrel$. The prop $\GQA$ is freely generated by the following operations and the equations in Figure~\ref{fig:gaa-axioms}, for $k \in \mathbb R$.
\begin{eqnarray*}
	\input{bcomult.tikz} \quad \input{bcounit.tikz} \quad \input{wmult.tikz} \quad \input{wunit.tikz}  \quad \scalar{k} \quad \input{foot.tikz} \quad \NN \quad
	\input{bmult.tikz} \quad \input{bunit.tikz} \quad \input{wcomult.tikz} \quad \input{wcounit.tikz} 
\end{eqnarray*}
\noindent We may regard $\GQA$ as an extension of $\GAA$ from Section~\ref{sec:GAA}, where we add $\NN$, and the associated equations. As we did for $\coscalar$ and $\input{cofoot.tikz}$ we can define the `mirror image' of $\NN$ as $\input{co-NN.tikz} :=  \input{co-NN-def.tikz}.$
We now define an interpretation of $\GQAfwd$-generators as morphisms of $\quadrel$. To each operation $c$ of type $m \to n$ we assign a quadratic relation $\sem{c}$ of the same type, i.e. a nonnegative partial quadratic function $\mathbb{R}^{m+n} \to [0,\infty]$. We make use of the notation $\iv{\phi}$ introduced in~\eqref{eq:iversonbrackets}. 
\begin{gather*}
		\sem{\input{bcomult.tikz}}\left({x},{x_1 \choose x_2}\right) = \iv{x =x_1 = x_2}
		\quad
		\sem{\input{bcounit.tikz}}\left({x}, ()\right) = 0 
		\quad
		\sem{\input{bunit.tikz} }\left((), {x}\right) = 0
		\\[7pt]
		\sem{\input{wmult.tikz}}\left({x_1 \choose x_2},{y}\right)= \iv{y = x_1+x_2} 
		\quad
		\sem{\input{wunit.tikz}}\left((), {x}\right) = \iv{x = 0}
		\quad
		\sem{\input{wcounit.tikz}}\left({x}, ()\right) = \iv{x = 0} 
		\\[7pt]
		\sem{\input{foot.tikz}}\left((), {x}\right)= \iv{x = 1} 
		\quad
		\sem{\scalar{k}}\left({x,y}\right)= \iv{y = k \cdot x} 
		\quad
		\sem{\NN}\left((), {x}\right) = \frac{1}{2}x^2
\end{gather*}
The interpretation for $\input{bmult.tikz}$, $\input{bunit.tikz}$, $\input{wcomult.tikz}$, and $\input{wcounit.tikz}$ is defined symmetrically, e.g. $\sem{\input{bmult.tikz}}\left({x_1 \choose x_2},{x}\right)= \iv{x =x_1 = x_2}$. Also, except for the new generator $\NN$, it extends conservatively the interpretation of these operations in $\affrel$ given in~\cite{pbsz} (see \Cref{prop:faithfulconservative} below). As for $\GQAfwd$, by freeness of $\GQA$ the interpretation yields a prop morphism  $\sem{\cdot}\colon \GQA \to \quadrel$.

\begin{example}\label{ex:compositionGQAdiag}
	Returning to \Cref{ex:compositionGQA}, $F = \sem{\Wcomult}$, $G = \sem{\input{doublecoNN.tikz}}$, and $\sem{\input{leastsquares_exampleLHS.tikz}} = F;G$. We can solve $(F;G)(c,()) = \inff{} {\frac 1 2 (x^2 + y^2) : x + y = c } = \frac 1 4 c^2$ by equational reasoning in $\GQA$. The derivation, which is the mirror image of the one given for $\input{add_example_3short.tikz}$ in \Cref{sec:intro}, proves  $\input{leastsquares_example.tikz}$ (observe that $\sem{\input{leastsquares_exampleRHS.tikz}}(c,()) = \frac 1 4 c^2$).
\end{example}

\subsection{Completeness of $\GQA$ for $\quadrel$}
\label{sec:gqa-completeness}

In this section we establish that $\GQA$ presents $\quadrel$. For this, we need to prove soundness, completeness and definability. 
Soundness is a routine check on the axioms (Proposition~\ref{prop:soundness-quadrel}). We focus on definability and completeness.

First, in a category where objects are equipped with Frobenius monoids, questions about arbitrary morphisms  can be reduced to questions about states, i.e. morphisms $0 \to n$, as follows: for a morphism $f : m \to n$, its `name' $\lceil f \rceil : 0 \to m \oplus n$ is defined as
\begin{equation} 
	\input{name.tikz} \label{eq:name} 
\end{equation}
\noindent The assignment $f \mapsto \lceil f \rceil$ is bijective -- every morphism can be recovered from its name by bending back the top wire, composing with $\input{blackcap-thick.tikz}$ and using axiom $\bullet$-\textsf{fr} and \textsf{coun}.

It will be useful to introduce a subprop $\GQAfwdex$ of $\GQA$, defined as follows. 
\begin{definition}\label{def:gqafwdex}
	The prop $\GQAfwdex$ is freely generated as $\GQAfwd$ with the addition of the generator $\input{bunit.tikz}$ and equations, for $k \in \mathbb{R}\!\setminus\! \{0\}$
	\begin{align*}
		\input{ax/bone-black.tikz}\; \myeq{} \; \input{empty-diag.tikz} \qquad \input{ax/TI-left.tikz} \;\myeq{TI}\;  \input{ax/TI-right.tikz} \qquad  \input{ax/bunitscalar.tikz}  \; \myeq{SI}  \; \input{bunit.tikz} \quad 
	\end{align*}
\end{definition}
Diagrams in $\GQAfwdex$ are not only a useful intermediary step in the proof of completeness; it turns out they axiomatise the \emph{extended Gaussians} of~\cite{stein2023gaussex,stein2023compositional}, which can be thought of intuitively as Gaussian maps that incorporate a form of (non-probabilistic) nondeterminism, modelled by the extra generator $\Bunit$ (see Appendix~\ref{app:extendendgauss} for details). In this paper, we focus on $\GQA$, and leave further exploration of $\GQAfwdex$ to future work.

A fundamental piece of the completeness proof involves showing that any quadratic problem, represented as a string diagram of $\GQA$ with no inputs, can be turned into a string diagram of the fragment $\GQAfwdex$. 
 This crucial step is Theorem~\ref{prop:cond_elimination} below, an elimination procedure which, given a state in $\GQA$, iteratively removes all occurrences of conditioning on some variable being zero ($\Wcounit$). 
The elimination procedure can therefore be understood as symbolically conditioning an (extended)
Gaussian to obtain an explicit posterior, or as solving a quadratic problem via QR decomposition.
\begin{restatable}{theorem}{condelimination}
	\label{prop:cond_elimination}
		Let $M : 0 \to n$ be a string diagram in $\GQA$. Then there exists a string diagram $M'$ in $\GQAfwdex$, and a `scalar' $\alpha : 0 \to 0$, such that $M = M' \oplus \alpha$ is derivable in $\GQA$. The scalar is of the form $\input{smallscalar.tikz}$ with $c \in [0,\infty)$, or $\input{infty.tikz}$.
\end{restatable}

\begin{theorem}\label{thm:presentation_quadrel}
	The prop $\quadrel$ is presented by $\GQA$.
\end{theorem}
\begin{proof}
	For definability, note that by taking names \eqref{eq:name}, it suffices to show that all states $M : 0 \to n$ (i.e. nonnegative partial quadratic functions) are definable. By \Cref{prop:quadrel_characterization}, we can always apply an affine change of coordinate, \emph{i.e.}, some invertible affine map $g$, such that  $g(\sem{M}(\vec{x_1},\vec{x_2},\vec{x_3})) = \frac 1 2 ||\vec{x_1}||^2 + \iv{\vec{x_3}=0}$, which is definable as the state $\raisebox{-0.25cm}{\input{elem_definability.tikz}}$. Hence, $\sem{M} = \input{definable.tikz}$.
	
	For completeness, it suffices to show that for all states $M_1,M_2 \colon 0 \to n$ of $\GQA$, if $\sem{M_1} = \sem{M_2}$ in $\quadrel$, then $M_1 = M_2$ is derivable in $\GQA$. The steps of the proof are as follows (see Appendix~\ref{sec:appendix_elimination} for details). First, using \Cref{prop:cond_elimination}, we transform each $M_i$ (for $i=1,2$) into the form $\input{scalars_oneline.tikz}$, where $M_i'$ lies in $\GQAfwdex$. Then, to each $M_i'$, we apply the normalisation procedure for states in $\GQAfwdex$ (\Cref{prop:nf_gaussex}), which brings them to the form $N_i :=\scalebox{0.8}{\input{gaussex_state_nf.tikz}}$, where $D_i$ is a vector subspace, $\mu_i \in {D_i}^\bot$, $L_i$ is lower triangular, and $\im(L_i {L_i}^T) \subseteq {D_i}^\bot$. By~\Cref{prop:gaussex_semantics_nf}), $\sem{N_1}=\sem{N_2}$ if and only if $D_1=D_2$, $\mu_1=\mu_2$, and $L_1L_1^T=L_2L_2^T$. Using \Cref{prop:encoding_uniqueness}, we can assume that $L_1 = L_2$ in the two diagrams. Then, by appealing to the completeness of $\GAA$, we can show that $N_1=N_2$. Finally, since $M_i=N_i$, we can conclude $M_1 = M_2$, as we wanted.
\end{proof}

\section{Applications}\label{sec:cases-studies}

\subsection{Functorial Transformations}\label{sec:functors}

A powerful application of our presentations for $\gauss$ and $\quadrel$ is the ease with which we can now define functors on these categories, as mappings on generators preserving the relevant equations. By the universal property of a presentation, such an assignment extends uniquely to a morphism of props. The first functor we consider clarifies the connection between Gaussian probability and quadratic problems. Relying on the isomorphisms $\gauss \cong \GQAfwd$ and $\quadrel \cong \GQA$, we can define a functor $L : \gauss \to \quadrel$ simply by mapping the generators of $\GQAfwd$ to the same generators in $\GQA$. Note that $L$ is functorial by construction, even though composition operations in $\gauss$ (integration/pushforward) and composition in $\quadrel$ (infimisation) are at first glance very different. It remains to compute how the functor $L$ acts concretely on a Gaussian map (see Appendix~\ref{app:functoriality} for details).
\begin{proposition}\label{prop:logdensity}
	The functor $L$ sends a Gaussian map $f(x) = Ax + \N(\mu, \Sigma)$ to the quadratic relation given by its negative conditional log-density. That is,
		$Lf(x,y)$ is defined as  $\frac 1 2 \langle (y-Ax - \mu),\Omega(y-Ax - \mu) \rangle + \iv{(y-Ax - \mu) \in \im(\Sigma)} $,
	where $\Omega$ is a generalised inverse of $\Sigma$, i.e. satisfies $\Sigma\Omega\Sigma = \Sigma$. 
\end{proposition}

\begin{proposition}\label{prop:faithfulconservative}
	The functor $L : \gauss \to \quadrel$ is faithful and $\GQA$ is conservative over $\GQAfwd$, meaning that $s = t$ in $\GQA$ implies $s = t$ in $\GQAfwd$.
\end{proposition}
\noindent These results provide an elegant and precise connection between the world of Gaussians and of convex optimisation problems. Functoriality of the log-density was laboriously proved in \cite{stein2023compositional} by explicit means. Here, we obtain the  result ``for free'' from our use of presentations, from which it also easily follows faithfulness and conservativity of $\GQA$ over $\GQAfwd$.


\noindent Another functor of interest, $S : \gauss \to \affrel$ can be defined by mapping every generator except $\NN$ to the same generator in $\affrel$, and mapping $\NN$ to $\Bunit$. This functor takes a Gaussian map to the affine relation given by its support; if $f(x) = Ax + \N(b,\Sigma)$, then $S(f) := \{ (x,y) : y \in Ax + b + \mathrm{im}(\Sigma) \}$.
Again, functoriality of this assignment is an immediate consequence of our presentation results. A similar definition yields a functor $\quadrel \to \affrel$, mapping $f: \R^{m + n} \to [0, \infty]$ to its effective domain $\text{dom } f = \{ (x, y) \mid f(x, y) < \infty \}.$

\subsection{Ordinary Least-Squares}\label{sec:ols}

	We now demonstrate how to apply our theory of $\GQA$ to the method of ordinary least squares in linear regression.
	The aim of linear regression is simple: to find a linear model that best fits a set of observations. In its usual vectorial formulation, all available observations of the regressors form the columns of a single matrix $A$ and all observations of the dependent variable form a single vector $y$; then, a linear model with parameters $x$ is expressed concisely as the system $Ax = y$. Typically, for consistency, we also assume that the regressors are linearly independent, \emph{i.e.}, that $A$ is injective. If $A$ is not invertible--as it usually is not--this system does not admit an exact solution. We can nevertheless look for parameters $x$ such that $A{x}$ best approximates the observed values $y$. Here, `best' is interpreted in such a way that the sum of squared errors, $||y-Ax||^2$, is minimised. This function can be translated directly into the following diagram: 
	$$
	\sem{\input{lq-problem.tikz}} \;= \; {x \choose y}\mapsto \frac{1}{2}||y-Ax||^2
	$$
	The formula for the optimal $\hat x$ is the familiar\ ordinary least squares (OLS) estimator $\hat x = (A^T A)^{-1}A^T y = A^+y$. This can be derived by applying equational reasoning in $\GQA$ (as detailed in Appendix~\ref{app:ols}), showing that
	$$
	\input{lq-problem.tikz}  \;=\; \input{ols-fin.tikz}
	$$
	The semantics of the last diagram is  ${x \choose y}\mapsto \frac{1}{2}\left(||AA^+y-Ax||^2 + ||y-AA^+y||^2\right)$.
	Its infimum is clearly reached at ${\hat x}=A^+y$, as wished: in this case $||AA^+y-A{\hat x}||^2=0$ and the remaining term $||y-AA^+y||^2 = ||y-A{\hat x}||^2$ indicates the distance between $A{\hat x}$ and $y$, \emph{i.e.}, how far we are from having successfully inverted~$A$.

\subsection{Gaussian Probabilistic Programming}\label{sec:ppl}

\noindent In \cite{stein2021compositional,stein2024} the authors study a simple first-order probabilistic programming language for Gaussian probability, which we name $\GPL$. The core probabilistic constructs are sampling from a standard-normal distribution \mlstinline{normal()} and an operator $(\eq)$ for conditioning two random variables to be equal. Terms of $\GPL$ are, for $\alpha,\beta \in \mathbb R$, $i \in \{0,1\}$,
\begin{align*}
	s,t ::= x \s s + t \s \alpha \cdot s \s \underline{\beta} \s (s,t) \s () \s \letin x s t \s s; t \s \pi_i\,s \s \normal() \s s \eq t 
\end{align*}
Terms are typed in random variables, see Appendix~\ref{app:typingjudg}. In $\GPL$ we can express an inference problem such as a noisy measurement as follows:
\begin{align}\label{eq:exampleprogram} 
\letin {x} {10 \cdot \normal()} \letin {y} {x + 5 \cdot \normal()} (y \eq \underline{40}); x
\end{align}
\noindent This expresses the following mathematical problem: If $X \sim \N(0,100)$ and $Y|X \sim \N(X,25)$, what is $X|(Y=40)$? The answer turns out to be $\N(32,20)$, which can be found by factorising the joint negative log-density: $f(x,40) \propto \frac 1 2 \frac{x^2}{100} + \frac 1 2 \frac{(x-40)^2}{25} = \frac 1 2 \frac{(x-32)^2}{20} + 6.4$.
The constant $c=6.4$ at the end corresponds to the \emph{score} of the problem, i.e. the negative logarithm of the \emph{normalisation constant} or \emph{model evidence}. The higher $c$, the less likely was the observation. 

It is possible to associate any $\GPL$ term to a $\GQA$-string diagram. This translation is rigorously definable as a functor, by understanding $\GPL$ as the internal language of a suitable monoidal category, as discussed in \cite{stein2021compositional,stein2021structural,di2024simple}. Given the focus of the present work, we do not pursue this systematically here, and confine ourselves to showing how the constructs of $\GPL$ map into $\GQA$:
\[\scalebox{.9}{\input{semantics.tikz}}\] 
For instance, the program~\eqref{eq:exampleprogram} translates to the leftmost diagram, and normalises modulo $\GQA$ to the posterior distribution and normalisation constant expressed by the rightmost diagram. 
\[ \scalebox{.9}{\input{semantics_example.tikz}} \]
Via this translation, $\GPL$ receives denotational semantics in $\quadrel$, which is comparable to existing $\GPL$ semantics in interesting ways: First, the elimination procedure of \Cref{prop:cond_elimination} mirrors the operational semantics of~\cite{stein2024} by reducing conditioning statements whenever they occur. Second, the normal form argument of \Cref{thm:presentation_quadrel} can be also employed to compare two $\GPL$ terms once understood as $\GQA$ diagrams. This may yield a decidability procedure for contextual equivalence, and thus a fully abstract (and completely axiomatised) denotational semantics for $\GPL$, which unlike the one in~\cite{stein2021compositional} does not rely on an equivalence relation that is hard to decide. A complete account of this conjecture is outside of the scope of this paper; we leave it for future work.

			\section{Conclusions}\label{sec:conclusions}
			
			We gave a compositional account of quadratic optimisation via a category of quadratic relations called  $\quadrel$, and provided a complete calculus $\GQA$ to reason algebraically about this class of problems. We also showed that a fragment $\GQAfwd$ of $\GQA$ axiomatises Gaussian probabilistic maps, thus highlighting the relationship between Gaussian probability and convex analysis. We illustrated our approach by discussing functorial tranlations between the Gaussian, quadratic, and affine domain, by modelling the method of ordinary least-squares in linear regression as diagrammatic reasoning, and by giving semantics to probabilistic programming. The last two applications are by no means intended to be the final word on the subject, but rather a starting point for a more systematic treatment. In particular, we plan to give a complete account of full abstraction for Gaussian probabilistic programming, as outlined at the end of \Cref{sec:cases-studies}. Also, we will further investigate the fragment $\GQAfwdex$ and its relationship with extended Gaussian relations. This semantic model is connected with the notion of open stochastic system in Willems' approach to control theory~\cite{willems:oss}, which we should be able to account for similarly to how $\GAA$ interprets electrical circuits (\cite{Boisseau-circuits,pbsz}). 
			
			
			From an algebraic viewpoint, we would like to explore whether the symmetries present in Graphical Affine Algebra extend to Graphical Quadratic Algebra. In particular, we conjecture the `colour-swap' symmetry, mapping eg. $\Bcomult$ to $\Wcomult$ and $\Wunit$ to $\Bunit$, extends to $\GQA$ by mapping $\NN$ to itself. 
			The semantics of this functor should correspond to convex conjugation (Legendre transformation). Indeed, in convex analysis terminology, every quadratic relation is a so-called convex bifunction (\cite{rockafellar}), which have been studied in categorical terms in \cite{stein2023compositional}. Independently and at the same time as a preprint~\cite{stein2024graphicalquadraticalgebra} of our work, the authors of \cite{Booth2024CompleteET} introduced a notion of `Gaussian relation' in the context of quantum computation; their semantics is very different, but their axiomatisation is highly reminiscent of $\GQA$. Their presentation exhibits the colour-swap symmetry discussed earlier, which is interpreted by the Fourier transform. The precise relationship between these approaches remains to be explored.
			
			\paragraph{Acknowledgements} F. Zanasi acknowledges support from \textsc{epsrc} grant EP/V002376/1, \textsc{miur} PRIN P2022HXNSC, and \textsc{aria} Safeguarded AI  programme.
			
			\bibliographystyle{splncs04}
			\bibliography{main}
			
			\appendix
			
			\section*{Appendix}\label{appendix}
			\renewcommand{\thesubsection}{\Alph{subsection}}
			
			\subsection{Laws of Symmetric Strict Monoidal Categories}\label[app]{app:lawssmc}
			
			\begin{equation*}
				\begin{gathered}
					{\input{smc/sequential-associativity.tikz} = \input{smc/sequential-associativity-1.tikz}}
					\\
					\scalebox{1}{\input{smc/unit-right.tikz} = \diagbox{c}{}{} = \input{smc/unit-left.tikz}}
					\\
					\scalebox{1}{\input{smc/parallel-associativity.tikz} = \input{smc/parallel-associativity-1.tikz}}
					\qquad
					\scalebox{1}{ \input{smc/parallel-unit-above.tikz} = \diagbox{c}{}{} =  \input{smc/parallel-unit-below.tikz}}
					\\
					\scalebox{1}{\input{smc/interchange-law.tikz} = \input{smc/interchange-law-1.tikz} }
					\\
					\scalebox{1}{\input{smc/sym-natural.tikz}= \input{smc/sym-natural-1.tikz}}
					\qquad\quad
					\scalebox{1}{\input{smc/sym-iso.tikz} = \input{smc/id2.tikz}}
				\end{gathered}
			\end{equation*}

			\subsection{Linear Algebra and Graphical Linear Algebra}\label[app]{app:linearalgebra}
			
			We recall some basic facts of linear algebra that will be useful in subsequent proofs. We will use classic matrix notation and string diagrammatic notation interchangeably: this is justified by  Proposition~\ref{prop:presentationsGLA}, which implies that any equivalence of matrices or subspaces represented by string diagrams are provable in the equational theory of $\GAA$.  

			A matrix $A \in \R^{n \times n}$ is \emph{orthogonal} if $AA^T = A^TA = I$. A \emph{rotation matrix} is a matrix $A$ that is orthogonal and furthermore satisfies $\det(A) = 1$. Any orthogonal matrix can be written as a product of a rotation matrix and a reflections. The classic matrix groups are defined as
			\begin{align*}
				\gl(n) &= \{ A \in \R^{n \times n} \text{ invertible } \} \\
				\orth(n) &= \{ A \in \gl(n) \text{ orthogonal} \}
			\end{align*}
			
			We recall the following well-known characterisations. Because $\GAA$ presents $\affrel$ (Proposition~\ref{prop:presentationsGLA}) and $\GQA$ is a conservative extension of $\GAA$, we may use them freely while doing diagrammatic reasoning in $\GQA$.
			
			\begin{proposition}[Column equivalence]\label{prop:column_equivalence}
				For two matrices $A,B \in \R^{n \times k}$, the following are equivalent
				\begin{enumerate}[topsep=0pt,itemsep=-1ex,partopsep=1ex,parsep=1ex]
					\item $\im(A) = \im(B)$
					\item there exists an invertible matrix $S$ with $AS=B$.
				\end{enumerate}
			\end{proposition}
			
			\begin{proposition}[Orthogonal column equivalence]\label{prop:orth_column_equivalence}
				For two matrices $A,B \in \R^{n \times k}$, the following are equivalent
				\begin{enumerate}[topsep=0pt,itemsep=-1ex,partopsep=1ex,parsep=1ex]
					\item $AA^T = BB^T$
					\item there exists an orthogonal matrix $R$ with $AR=B$.
				\end{enumerate}
			\end{proposition}
			
			Two subspaces $S,D \subseteq \R^n$ are \emph{complementary} if $S + D = \R^n$ and $S \cap D = 0$. Every vector $x \in R^n$ decomposes uniquely as $x = x_S + x_D$ with $x_S \in S, x_D \in D$. We obtain the corresponding projections $P_S, P_D$. A canonical choice of complementary subspace is the orthogonal complement $S=D^\bot$.
			
			\begin{proposition}[Invariance]\label{prop:gl_invariance}
				$\input{bunit.tikz}$ is invariant under $\gl(n)$, i.e. for all invertible matrices $S \in \gl(n)$ we have in $\GAA$
				\begin{equation*} \input{gli.tikz} \tag{GLI} \end{equation*}
			\end{proposition}
			\begin{proof}
				The Gaussian elimination algorithm, which we may mimic by equational reasoning in $\GAA$, demonstrates that every invertible matrix can be built up from three kinds of elementary transformations
				\begin{enumerate}
					\item scaling $x \mapsto \alpha x$, with $\alpha \neq 0$
					\item swapping $(x,y) \mapsto (y,x)$
					\item the shear mapping $(x,y) \mapsto (x,kx + y)$
				\end{enumerate}
				The generator $\Bunit$ is invariant under swapping by the laws of symmetric monoidal categories (Appendix~\ref{app:lawssmc}), and under scaling by (SI). It remains to show that it is invariant under the shear mapping:
				\begin{equation*} \input{gli_shear.tikz} \end{equation*}
			\end{proof}
			
			\begin{proposition}[Uniqueness of subspace encoding]\label{prop:subspace_encoding_uniqueness}
				Let $A \in \R^{n \times m_1}, B \in \R^{n \times m_2}$ be two matrices satisfying $\im(A) = \im(B)$. Then we can derive in $\GAA$
				\[ \input{subspace_state_A.tikz} = \input{subspace_state_B.tikz} \]
			\end{proposition}
			\begin{proof}
				Without loss of generality we can assume that $m_1 = m_2$: this is because if, say, $m_1 > m_2$, we can extend $B$ with further zero columns to $\widetilde B = (B\,0)$ and this still satisfies the assumptions \input{subspace_eq_tilde.tikz} of the proposition. By column equivalence (\Cref{prop:column_equivalence}), $AR = B$ for some invertible matrix $R$. Then the desired equality follows from the fact that $\Bunit$ is invariant under invertible matrices (\Cref{prop:gl_invariance}).
			\end{proof}
			
			\begin{proposition}\label{prop:projection_formula}
				If $S,D$ are complementary subspaces, then we can prove in $\GAA$ that
				\[ \input{projection_formula.tikz} \]
				where $P_S,P_D$ are the projection operators onto the subspaces.
			\end{proposition}
			\begin{proof}
				Using the fact that $P_S + P_D = I$, we can derive in $\GAA$.
				\[ \input{projection_formula_proof.tikz} \]
				In the step $(\ast)$ we use the fact that matrices are linear (commute with $\Wmult$), which is provable inside $\GAA$.
			\end{proof}

			\subsection{Soundness Proofs}\label[app]{app:soudness}
			
			\begin{proposition}\label{prop:soundnessgauss}
				The interpretation of $\GQAfwd$ in $\gauss$ is sound.
			\end{proposition}
			\begin{proof}
				The equations in the fragment corresponding to $\GAAfwd$ are evident because $\aff$ embeds in $\gauss$. It remains to show the validity of (RI) and (D). Let $R \in O(2)$, then we have
				\begin{align*}
					\semG{\input{gauss_soundness_1.tikz}} &= \N(0,RR^T) = \N(0,I) = \semG{\input{gauss_soundness_2.tikz}} \\
					\semG{\NN \! \Bcounit} &= \N(0,1) ; ! = \semG{\input{empty-diag.tikz}}
				\end{align*}
			\end{proof}
			
			\begin{proposition}
				\label{prop:soundness-quadrel}
				The interpretation of $\GQA$ in $\quadrel$ is sound.
			\end{proposition}
			\begin{proof}
				$\affrel$ embeds in $\quadrel$ via indicator functions of affine relations. Thus all the axioms of $\GAA$ are verified. It remains to verify the three axioms involving the quadratic generator, namely (RI), (D) and (Z).
				\begin{align*}
					\sem{\input{gauss_soundness_1.tikz}}(\vec y) &= \inff {\vec x} {\frac 1 2 ||\vec x||^2 + \iv{\vec y=R\vec x}} \\
					&= \frac 1 2 ||R^{-1}\vec y||^2 = \frac 1 2 ||\vec y||^2 \\
					&= \sem{\input{gauss_soundness_2.tikz}}(\vec y) \\
					\sem{\NN \! \Bcounit} &= \inff {x} {\frac 1 2 x^2 + 0} = 0 = \sem{\input{empty-diag.tikz}} \\
					\sem{\NN\!\Wcounit} &=	\inff {x} {\frac 1 2 x^2 + \iv{x=0}} = \frac 1 2 \cdot 0^2 = \sem{\input{empty-diag.tikz}}
				\end{align*}
			\end{proof}
			
			\subsection{Completeness of Causal GQA}\label[app]{app:completeness-causal-gqa}
			
			\begin{proposition}\label{prop:flip_invariance}
				$\NN$ is flip invariant, i.e.
				\[ \input{nn_flip_stm.tikz} \]
			\end{proposition}
			\begin{proof}
				By discardability and rotation invariance for $\varphi = \pi$,
				\[ \input{nn_flip_proof.tikz} \]
			\end{proof}

We show here that the fragment $\GQAfwd$, as introduced in~\Cref{sec:causal-GQA}, presents $\gauss$. The key is a normal form argument, which builds on the following propositions. 
The first proposition generalises axiom \textsf{RI} to arbitrary orthogonal matrices and the second will allow us to derive the uniqueness of our normal form. 

\begin{proposition}
	\label{prop:orthogonalinv}
	For any $n\times n$ orthogonal matrix $R$ we can derive in $\GQAfwd$ that $\input{oi.tikz} = \input{nn-thick.tikz}$. 
\end{proposition}
			\begin{proof}
				A \emph{Givens rotation} is a rotation matrix which acts as a rotation along two coordinate axes, and the identity otherwise. From (RI), it follows immediately that $\NN$ is invariant under Givens rotations. Now we use the fact (proved below) that every orthogonal matrix $R$ can be written as a product $R = Q_1\cdots Q_n \cdot D$ where the $Q_i$ are Givens rotations and $D$ is a diagonal matrix with entries $\pm 1$. From this, the proposition follows by repeated application of invariance under Givens rotations and flip invariance (\Cref{prop:flip_invariance}).
				
				It remains to prove the claimed decomposition of $R$, which we establish as follows: The usual algorithm for QR-decomposition gives $R=QU$, where $Q$ is a product of Givens rotations and $U$ is upper triangular. Then $U=Q^TR$ is orthogonal (normal) in addition to being upper triangular, hence must be diagonal, with entries $\pm 1$. All of these linear algebraic transformations can be performed diagrammatically in $\GQAfwd$, by Proposition~\ref{prop:presentationsGLA}.
			\end{proof}
			
\begin{proposition}\label{prop:encoding_uniqueness}
	For $A \in \R^{n \times m_1}, B \in \R^{n \times m_2}$ such that $AA^T = BB^T$,  $\input{gauss_state_A.tikz} = \input{gauss_state_B.tikz}$.
\end{proposition}
			\begin{proof}
				Without loss of generality we can assume that $m_1 = m_2 = m$: This is because if, say, $m_1 > m_2$, we can extend $B$ with further zero columns to $\widetilde B = (B\,0)$ and this still satisfies the assumptions $AA^T = \widetilde B\widetilde B^T$ and 
				$$
				\input{gauss_state_eq_tilde.tikz}
				$$
				By \Cref{prop:orth_column_equivalence}, $A = BR$ for some $R \in \orth(m)$, hence $\input{gauss_state_A.tikz} = \input{gauss_state_B.tikz}$ follows from Proposition~\ref{prop:orthogonalinv}.
			\end{proof}

\noindent We can now derive an encoding of Gaussian distributions in $\GQAfwd$.

			\begin{definition}
	Let $\mathcal N(\mu,\Sigma)$ be a Gaussian distribution on $\R^n$. For $\Sigma = LL^T$, where $L$ is lower triangular, we define $\input{gauss_encoding.tikz}$.
\end{definition}

			\begin{proposition}\label{prop:nf_gauss}
				Let $M : m \to n$ be any string diagram in $\GQAfwd$. Then $M$ can be brought into the form
				\[ \input{gauss_definability.tikz}\]
				with $A \in \R^{n \times m}, \mu \in \R^n$ and $\Sigma \in \R^{n \times n}$ positive semidefinite.
			\end{proposition}
			\begin{proof} By the laws of symmetric monoidal categories, (Appendix \ref{app:lawssmc}) it is possible to move all the occurrences of the generators $\NN \colon 0 \to 1$ to the extreme left of any given string diagram where they occur, while preserving equality in $\GQAfwd$. Applying this to $M$, we obtain a string diagram of shape $\input{tildeM.tikz}$ equivalent to $M$, where $\widetilde{M}$ lies in the affine fragment $\aff$ of $\GQAfwd$, meaning that $\semG{\widetilde{M}}$ is given by $f(x_1,x_2) = Lx_1 + \mu + Ax_2$ for some matrices $A,L \in \mathbb{R}^{n\times m}$ and $\mu \in \mathbb{R}^n$. Now, $\semG{\widetilde{M}}$ is also a morphism of $\aff$: by definition of $\sem{G}{\widetilde{M}}$ as above and completeness of $\GAAfwd$ for the affine fragment~\cite{pbsz}, we have that
			\[ \input{nf_gauss_2bis.tikz} \]
			is provable in $\GQAfwd$, and thus in $\GQAfwd$. We can then conclude by observing that, by definition of composition in $\gauss$ and functoriality of $\semG{\cdot}$:
			\[ \input{nf_gauss_2.tikz} \]
			\end{proof}
			
			\subsection{Quadratic Relations}\label[app]{app:quadraticproofs}
			We give the proofs missing in~\Cref{sec:backgroundQuadraticRelations}.
			We will use the following characterisation.
\begin{proposition}\label{prop:quadrel_characterization}
	The following are equivalent for a function $f : \R^n \to [-\infty,\infty]$
	\begin{enumerate}[topsep=0pt,itemsep=-1ex,partopsep=1ex,parsep=1ex]
		\item $f$ is a nonnegative partial quadratic function
		\item\label{it:sylvester} $f$ can be written in the form $f(x) = h(A(x-a)) + c$
		where $A$ is an invertible matrix, $c \geq 0$ and $h$ an elementary convex partial quadratic function
		\item $f$ can be written as $f(x) = \inff {} { ||y||^2 : (x,y) \in M }$
		where $M \subseteq \R^n \times \R^m$ is an affine subspace.
	\end{enumerate}
\end{proposition}
Let us first describe the meaning of these statements. \Cref{it:sylvester} is an analogue of Sylvester's law of inertia for partial quadratic functions; in fact, it suffices to allow the values $\lambda_i \in \{0,1,\infty\}$ on the diagonal. Item 3 shows that nonnegative partial convex functions are precisely the class of functions that arise in constrained least squares optimisation~\cite{Boydbook}, which justifies our choice of using them as the morphisms of the semantic domain $\quadrel$. 
			
			\begin{proof}[Proof of \Cref{prop:quadrel_characterization}]
				$(1) \Rightarrow (2)$ using an affine coordinate change, we can assume that the domain of $f$ is the subspace $M = \R^m \times \{0\}$. The restricted function $\tilde f(x_1, \ldots, x_m) = f(x_1,\ldots,x_m,0,\ldots,0)$ is a total nonnegative quadratic function so we can apply the characterisation \Cref{prop:quadrel_characterization} to $\tilde f$ to obtain the desired form
				\[ f(x) = h(A((x_1,\ldots,x_m)-a)) + \infty \cdot x_{m+1} + \ldots + \infty \cdot x_n + c\]
				
				$(2) \Rightarrow (3)$ Under the affine coordinate change $\vec{y} = A(\vec x - \vec a)$, it suffices to consider $f(\vec y) = h(\vec y) + c$ where $h$ is an elementary quadratic function of the form $h(\vec{y_1},\vec{y_2},\vec{y_3}) = ||\vec{y_1}||^2 + \iv{\vec{y_3} = 0}$. In this case, we can define the affine relation
				\[ M = \{ ((\vec{y_1},\vec{y_2},\vec{y_3}),(\vec{y_1},\sqrt c)) \s \vec{y_3} = 0 \} \]
				and have $f(x) = \inff {} {||z||^2 : (x,z) \in M}$.
				
				$(3) \Rightarrow (1)$ If $M \subseteq \R^n \times \R^m$ is an affine relation, then there exists a vector subspace $D \subseteq \R^n$ and an affine map $g : \R^n \to D^\bot$ such that
				\[ M_x = \begin{cases}
					g(x) + D, &x \in \pi_X M \\
					\emptyset, &\text{otherwise}
				\end{cases}\]
				where $M_X = \{ y : (x,y) \in M\}$. Thus we have
				\[ \inff {} {||y||^2 : y \in M_x} = ||g(x)||^2 + \iv{x \in \pi_XM} \]
				which is evidently a nonnegative partial quadratic function.
			\end{proof}
		
		We are now ready for the proof of \Cref{lemma:quadrel_infimization}.
		
			\begin{proof}[Proof of \Cref{lemma:quadrel_infimization}]
				Let $g : \R^m \to \exR$ be a nonnegative partial quadratic function and $M \subseteq \R^m \times \R^n$ an affine relation. Using the representation of \Cref{prop:quadrel_characterization} and an affine coordinate change, it suffices to consider the case where $g$ is an elementary function, i.e. a constrained minimization problem of the form
				\[ f(x) = \inff {} { ||\vec{y_1}||^2 + \iv{\vec{y_3}=0} : (x,(\vec{y_1},\vec{y_2},\vec{y_3})) \in M } \]
				We define the affine relation $M' = \{ (x,\vec{y_1}) : (x,(\vec{y_1},\vec{y_2},0)) \in M \}$ and conclude that $f(x) = \inff {} { ||y||^2 : (x,y) \in M' }$ which is a partial quadratic function by \Cref{prop:quadrel_characterization}.
			\end{proof}

			\subsection{Completeness for Extended Gaussian Processes}\label[app]{app:extendendgauss}
			
			An \emph{extended Gaussian distribution} is a Gaussian distribution on a quotient space $\R^n/D$; the vector subspace $D$ is called a \emph{nondeterministic fibre}. This concept intends to model the idea of a distribution carrying both a probabilistic and a nondeterministic component. We can write it in additive notation $\N(\mu,\Sigma) + D$, where we think of the subspace $D$ as an idealised uniform distribution. Extended Gaussians support the same basic transformations as ordinary Gaussians, such as pushforwards and conditioning~\cite{stein2023gaussex}. A related notion appears in Willems' approach to control theory, under the name of open stochastic system~\cite{willems:oss}. A major reason of interest for extended Gaussians is that they give semantics to Gaussian probabilistic programming extended with a construct $\uniform()$ to sample from a uniform prior~\cite{stein2023gaussex,stein2023compositional}. 
			
			Similarly to Gaussian maps (\Cref{sec:gauss}), one may define a notion of extended Gaussian map, informally written as $f(x) = Ax + \N(\mu,\Sigma) + D$, where $Ax + \N(\mu,\Sigma)$ is a Gaussian map and $D$ is the nondeterministic component. A prop $\gaussex$ of these maps is defined~\cite{stein2023gaussex}. 
			
			
			We first recall the complete definition of the prop $\gaussex$ of extended Gaussian maps.
			
			\begin{definition}[\cite{stein2023gaussex}]
				The prop $\gaussex$ has morphisms $m \to n$ pairs $(m \xrightarrow{A} k \xleftarrow{P} n, \varphi)$, where $m \xrightarrow{A} k \xleftarrow{P} n$ is a cospan in $\vect$, $P$ is surjective, and $\varphi$ is a morphism of type $0 \to k$ in $\gauss$, i.e. a Gaussian distribution $\N(\mu,\Sigma)$ on $\mathbb{R}^k$. In additive notation, we may represent such a morphism as $f(x) = Ax + \N(\mu',\Sigma') + D$, where the kernel subspace $D := \ker(P)$ is the nondeterministic fibre, and $\mu' := Q\mu$, $\Sigma' := Q \Sigma Q^T$ for some right-inverse $Q$ of $P$. Note $f(x)$ selects one element $z \in D$ of the fibre, which we can interpret as the non-deterministic choice of $f$.
				
				The monoidal product of $(m \xrightarrow{A} k \xleftarrow{P} n, \varphi)$ and $(m' \xrightarrow{A'} k' \xleftarrow{P'} n', \varphi')$ is defined as $(m+m' \xrightarrow{A\oplus A'} k+k' \xleftarrow{P\oplus P'} n+n', \varphi \oplus \varphi')$, using the monoidal product in $\vect$ on the cospan legs and the one in $\gauss$ on the distributions. The sequential composition of $(m_1 \xrightarrow{A_1} k_1 \xleftarrow{P_1} n, \varphi_1)$ and $(n \xrightarrow{A_2} k_2 \xleftarrow{P_2} n_2, \varphi_2)$ is defined as $(m_1 \xrightarrow{A_1 ; L} k \xleftarrow{P_2;R } n_2, \varphi)$, where $k_1 \xrightarrow{L} k \xleftarrow{R} k_2$ is the pushout of $k_1 \xleftarrow{P_1} n \xrightarrow{A_1} k_2$ in $\vect$,
				in $\vect$ and $\varphi$ is the convolution in $\gauss$ of distributions $\varphi_1 ; L \colon 0 \to k$ and $\varphi_2 ; R \colon 0 \to k$.\footnote{As usual, for composition by pushout to be well-defined, strictly speaking morphisms of this category are equivalence classes of isomorphic cospans.}
			\end{definition}
			One may phrase $\gaussex$ more abstractly as a category of `decorated' cospans~\cite{fong2015decorated}, see \cite{stein2023gaussex} for details. Also, $\gaussex$ is a Markov category~\cite{fritz-markov}. Note there is an embedding $\gauss \to \gaussex$ where $(A,b,\Sigma)$ is sent to the cospan $(m \xrightarrow{A} n \xleftarrow{Id} n, \N(b,\Sigma))$ whose right leg is the identity, i.e. the nondeterministic fibre vanishes. Also there is an embedding $\gaussex \to \quadrel$ mapping the cospan $(m \xrightarrow{A} k \xleftarrow{P} n, \N(b,\Sigma))$ to the negative log-likelihood function
			\begin{equation*}
				\begin{split}
					f(x, y) = \langle Py - Ax - b, \Sigma^+ (Py - Ax - b) \rangle + \\ \iv{Py - Ax - b \in \im(\Sigma)}.
				\end{split}
			\end{equation*}
			with $\Sigma^+$ the Moore-Penrose pseudoinverse of $\Sigma$. This embedding is discussed in~\cite{stein2023compositional}. Proving functoriality is highly non-trivial, but it becomes a simple corrollary when paired with our presentation results, see \Cref{sec:functors}.

						\begin{proposition}\label{prop:nf_gaussex}
				Let $M : m \to n$ be any string diagram in $\GQAfwdex$. Then $M$ can be brought into the form
				\[ \input{nf_gaussex.tikz}\]
				where $D \subseteq \R^n$ is a vector subspace, and $\mu \in D^\bot, \im(\Sigma) \subseteq D^\bot, \im(A) \subseteq D^\bot$.
			\end{proposition}
			\begin{proof}
				The proof is analogous to the one of \Cref{prop:nf_gauss}. Given a string diagram $M$ in $\GQAfwdex$, we may move all the occurrences of $\NN$ and $\Bunit$ to the left, thus obtaining a string diagram $\widetilde{M}$ in $\GAAfwd$ such that $M =\input{bunitxnn-tildeM.tikz}$, where $\semG{\widetilde{M}}$ is given by the linear map $f(\vec{x_1},\vec{x_2},\vec{x_3})) = S\vec{x_1} + M\vec{x_2} + \nu + B\vec{x_3}$ for some matrices $B,M \in \mathbb{R}^{n\times m}$ and vector $\nu \in \mathbb{R}^n$. Appealing to the completeness of $\GAAfwd$ for the affine fragment~\cite{pbsz}, we can obtain the following diagram:
				\[M = \input{bunitxnn-tildeM.tikz} =  \input{nf_gaussex_2.tikz} \]
				Then, let $D := \im(S)$. By the associativity of $\Wmult$ followed by \Cref{prop:projection_formula}, we get 
				\[
				\input{nf_gaussex_2.tikz} = \input{nf_gaussex_3.tikz} = \input{nf_gaussex_4.tikz} 
				\]
				Then, since matrices distribute over the sum ($\Wmult$) we have 
				\[
				 \input{nf_gaussex_4.tikz}  = \input{nf_gaussex_5.tikz} = \input{nf_gaussex_6.tikz} 
				\]
				where $P_S$ is the orthogonal projection onto $D^\bot$, and $A:= P_S B$, $L:= P_SM$, $\mu = P_S\nu$. Finally, with $\Sigma := LL^T$, we get
				\[ \input{nf_gaussex_6.tikz} =  \input{nf_gaussex.tikz} \]
			\end{proof}
			
			Using similar methods as in \Cref{sec:causal-GQA},we can now show that the prop $\GQAfwdex$ presents $\gaussex$.
						\begin{theorem}\label{thm:completenessgaussex}
				The prop $\gaussex$ is presented by $\GQAfwdex$.
			\end{theorem}
			An intuitive reading of Theorem~\ref{thm:completenessgaussex} is that the extra generator $\input{bunit.tikz}$ uniformly `outputs' values in $\mathbb{R}$, and thus suffices to model the form of nondeterminism expressed by the component $D$ in the definition of extended Gaussians. 
			
			First, we need to define an interpretation $\semEG{\cdot}$ of the generators of $\GQAfwdex$ in $\gaussex$:
			\begin{eqnarray*}
				\semEG{\input{bunit.tikz}}
				& := &  \left( 0 \xrightarrow{Id} 0 \xleftarrow{!} 1 \ , \ \N \left( \nullm,
				\zeromat \right) \right) \\
				\semEG{\input{generic-gen.tikz} \colon m \to n} & := &  \left( m \xrightarrow{A} n \xleftarrow{Id} n \ , \ \N(b,\Sigma) \right)  \qquad  \\[3pt]
			\end{eqnarray*}
			for any $\GQAfwd$-generator \input{generic-gen.tikz} with $\semG{\input{generic-gen.tikz}} \ = \ x \mapsto Ax + \N(b,\Sigma)$, where $!: 1 \to 0$ is the only linear map $x \mapsto \nullm$. This interpretation extends to a prop morphism $\GQAfwdex \to \gaussex$, which we also write $\semEG{\cdot}$.
			
			We now turn attention to the proof of \Cref{thm:completenessgaussex}.

			\begin{proof}[Proof of \Cref{thm:completenessgaussex}] The proof consists in verifying definability, soundness, and completeness under the prop morphism $\semEG{\cdot}$.
				First we show definability. Let $f = (m \xrightarrow{A} k \xleftarrow{P} n, \N(\mu,\Sigma))$ be a morphism of $\gaussex$, informally representing the map $f(x) = Ax + \N(\mu,\Sigma) + D$ where $D = \ker(P)$. Then $f$ is definable by the term
				\begin{equation} \input{nf_gaussex.tikz} \label{eq:gaussex_normal_form} \end{equation}
				For completeness, we note that the form \eqref{eq:gaussex_normal_form} is not quite a normal for $\gaussex$ semantics, but it becomes one if we insist on the further conditions that
				\[ \mu \in D^\bot, \im(A) \subseteq D^\bot, \im(\Sigma) \subseteq D^\bot \]
				where $D^\bot \subseteq \R^n$ is the orthogonal complement of $D$.
				As detailed in \Cref{prop:nf_gaussex}, we can transform every string diagram of $\GQAfwdex$ into that specific shape. The values of $D,A,\mu,\Sigma$ are then uniquely determined from the interpretation $\semEG{\cdot}$ in $\gaussex$.
			\end{proof}
			

			\subsection{Completeness of GQA}
			\label[app]{app:completeness-gqa}

			\subsubsection{Scalars}\label[app]{app:scalars}
			
			In this section we study properties of scalars ($0 \to 0$ morphisms of $\GQA$). We use the following shorthand notation:
			\[\input{blackcup-nonode.tikz} := \input{blackcup.tikz} \qquad \input{whitecup-nonode.tikz} := \input{whitecup.tikz} \]
			
			\begin{proposition}
				For $a,b,c \in [0,\infty)$, we have
				\[ \input{scalars_equal.tikz} \]
				whenever $a^2 + b^2 = c^2$.
			\end{proposition}
			\begin{proof}
				The first equation follows from flip invariance of $\NN$ (Proposition~\ref{prop:flip_invariance}). For the second equation, we can find a rotation matrix $R \in \R^{2 \times 2}$ which sends the vector $(a,b)$ to $(c,0)$. Then by (RI) and (Z)
				\[ \input{scalars_equal_proof.tikz} \]
			\end{proof}
			
			\begin{proposition}
				The equation
				\[ \input{scalars_nf.tikz} \]
				is derivable in $\GQA$ if and only if $a^2 = b^2$.
			\end{proposition}
			\begin{proof}
				If $a^2 = b^2$, then $a = \pm b$, and the equation can be derived from flip invariance of $\NN$ (Proposition~\ref{prop:flip_invariance}). In $\catname{QuadRel}$, their denotations are given by $\frac 1 2 a^2, \frac 1 2 b^2$ respectively.
			\end{proof}
			The initialisation principle states that if $X \sim \N(0,1)$ and we condition $X$ to be equal to a constant $c$, the posterior distribution is precisely $c$ up to a scalar.
			\begin{proposition}[Initialisation Principle] \label{prop:init-principle}
				The following schema is derivable for all $c \in [0,\infty)$.
				\begin{equation} \input{ini.tikz} \label{eq:ini} \tag{INI} \end{equation}
				If (NORM) holds, this simplifies further to
				\begin{equation*} \input{ini_norm.tikz} \end{equation*}
			\end{proposition}
			\begin{proof}
				This is given by the derivation
				\[ \input{ini_proof.tikz} \]
				where in the first step we use the black Frobenius algebra structure, and in the second step we use axioms (1-dup) and (dup) to cancel $\input{bcomult.tikz}$.
			\end{proof}
			
			\subsubsection{Elimination Procedure for Conditioning}\label[app]{sec:appendix_elimination}
			
			Let $M : 0 \to n$ be a state in $\GQA$. We give an inductive procedure to remove occurrences of the conditioning effect $\input{wcounit.tikz}$, except possibly in the scalar $\infty$.
			
			\condelimination*
			\begin{proof}
				We proceed by induction over the number $K$ of generators $\input{wcounit.tikz}$ appearing in $M$. First, by the laws of symmetric monoidal categories, (Appendix \ref{app:lawssmc}) it is possible to move all the occurrences of the generators of type $0 \to 1$ ($\NN$, $\Bunit$, and $\input{foot.tikz}$) to the extreme left of any given string diagram of $\GQA$ where they occur, while preserving equality of string diagram in the equational theory. Similarly, we can move all occurrences of $\input{wcounit.tikz}$ to the right. The result of applying this to $M$ is displayed in the first equality below. The remaining string diagram $\underline{A}$ is necessarily in the affine fragment of $\GQA$, its semantics being some matrix $A$. By applying equational reasoning in $\GQA$ (in fact, in its fragment $\GAAfwd$), which is complete for $\aff$~\cite{pbsz}, we can then separate $A$ into blocks, obtaining a diagram for its first row $A_1$, and one for the remaining rows $A'$ as on the right-hand side of the second equality below. Finally, we can break down $A_1$ further into row vectors $\alpha$, $\beta$, and $\gamma$, for each of the columns corresponding to the three incoming wires, as on the right-hand side of the third equality below:
				\[ \input{nf_gaussrel.tikz}\]
				where $N'$ is the diagram in the dotted box above, which contains $K-1$ occurrences of $\input{wcounit.tikz}$.
				
				We now distinguish two cases. If in $\alpha$ there is a coefficient $\alpha_i \neq 0$, we eliminate $\input{wcounit.tikz}$ using equational reasoning in $\GAAfwd$, as follows:
				\[ \input{rw_alpha.tikz}\]
				Note that this step at the same time eliminates one occurrence of the generator $\Bunit$.
				In the remaining case, all $\alpha_i$s are equal to $0$, then by equational reasoning in $\GQA$ we can discard $\alpha$ and the wire going into it. We can simplify the string diagram further, first by repeatedly applying axiom (1-dup) to $\input{foot.tikz}$, then by including the remaining occurrences of $\Bunit$ in $N'$, which yields $N$. This brings to a string diagram in the shape
				\[ \input{nf_gaussrel_2.tikz}\]
				with $c \in \R$. We can now find an orthogonal matrix $R \in \orth(n)$ with $\beta R = (b,\ldots,0)$ where $b = ||\beta||$:
				\begin{equation}\label{eq:elimanation-procedure}
					\input{nf_gaussrel_2.tikz} \quad\myeq{RI} \quad \input{nf_gaussrel_RI.tikz}
				\end{equation}
				where we have used the fact that matrices distribute over $\Bcomult$ (\emph{i.e.}, can be copied). Now, since $\beta R= (b,0,\dots,0)$ there are two further cases to consider.
				\begin{itemize}
					\item If $b > 0$, we can simplify the rightmost diagram of~\eqref{eq:elimanation-procedure} further, with the matrix $\beta R = (b,0,\dots,0)$ reducing to the string diagram in the dotted box in the first step below (\emph{cf.}~\eqref{eq:matrixform}). Then, using 
					(INI) as defined in Proposition~\ref{prop:init-principle}, we get:
					\[ \input{nf_gaussrel_3.tikz}\]
					where the last dotted frame identifies the diagram $M'$ that we wanted to prove the statement of the theorem.
					\item If $b=0$, the rightmost string diagram of \eqref{eq:elimanation-procedure} instead simplifies as
					\[ \input{nf_gaussrel_4.tikz}\]
					where the scalar is either $0$ or $\infty$ depending on whether $c=0$.
				\end{itemize}
				Then the scalars are normalised separately according to the procedure of Appendix~\ref{app:scalars}.
			\end{proof}
			
			We now give further lemmas needed for the proof of completeness of $\GQA$.
			
			First, recall that a \emph{generalised inverse} of a matrix $A \in \R^{n \times m}$ is a matrix $G \in \R^{m \times n}$ such that $AGA=A$; that is, for all $y \in \im(A)$ we have that $Gy$ is a preimage of $y$. Every matrix has a generalised inverse, but it is generally not unique as there are two degrees of freedom: What shall $G$ do outside of the image $\im(A)$, and which preimage of $y$ should it pick? The \emph{Moore-Penrose pseudoinverse} $A^+$ is a particularly well-known generalized inverse which answers both questions by projecting orthogonally.
			
			The Moore-Penrose pseudoinverse gives an explicit solution to least-squares problem:
			\[ \inf \{ ||x||^2 : Ax = y \} \]
			is attained at $x=A^+y$ when $y \in \im(A)$.
			
			\begin{proposition}\label{prop:sem-scalar}
				Let $f : m \to n$ be a quadratic relation and $c \in \R^n$, then	\[ \sem{\input{interp_shift_ct.tikz}}(x,y) = f(x,y-c) \]
			\end{proposition}
			\begin{proof}
				We compute
				\[ \inff {y_1,y_2} { f(x,y_1) + \iv{y_2 = c} + \iv{y=y_1+y_2}} = f(x,y-c) \]
			\end{proof}
			
			\begin{proposition}\label{prop:projection_formula_semantics}
				Let $S,D$ be complementary subspaces, and let $f : 0 \to n$ be such that $f(x)= \infty$ when $x \notin S$. Then
				\[ \sem{\input{projection_formula_2.tikz}}(x) = f(x_S) \]
			\end{proposition}
			\begin{proof}
				When taking the infimum $\inff {x_1,x_2} {f(x_1)+\iv{x_2 \in D} + \iv{x_1+x_2 = x}}$, the conditions $x_1 \in S, x_2 \in D$ force $x_1 = x_S, x_2 = x_D$ by uniqueness of the decomposition.
			\end{proof}

			\begin{proposition}
				In $\quadrel$, we have the following interpretations for each constant $c \in \R^n$, matrix $A \in \R^{n \times m}$ and affine subspace $W \subseteq \R^n$
				\begin{align*}
					\sem{\const{c}}(x) = \iv{x=c} \qquad
					\sem{\mat{A}}(x,y) = \iv{y=Ax} \qquad
					\sem{\subspace{W}}(x) = \iv{x \in W}.
				\end{align*}
			\end{proposition}
			
			\begin{proposition}\label{prop:interp-gauss}
				Let $A \in \R^{n \times m}$, then
				\[ \sem{\input{interp_gauss.tikz}}(y) = \frac 1 2 ||A^+y||^2 + \iv{y \in \im(A)} \]
				where $A^+$ is the Moore-Penrose pseudoinverse of $A$. This expression is furthermore equal to
				\[ \frac 1 2 \langle y,\Omega y \rangle + \iv{y \in \im(A)} \]
				where $\Omega$ is any generalized inverse of $AA^T$.
			\end{proposition}
			\begin{proof}
				The denotation corresponds to the quadratic problem
				\[ \inff x { \frac 1 2 ||x||^2 + \iv{y=Ax} } = \inff {} { \frac 1 2 ||x||^2 : y = Ax } \]
				which is solved using the Moore-Penrose pseudoinverse. We have
				\[ \langle A^+y,A^+y \rangle = \langle y, (A^+)^TA^+y \rangle = \langle y,\Omega y \rangle \]
			\end{proof}
			
						As a Corollary of~\Cref{prop:interp-gauss} and \Cref{prop:sem-scalar}, we can compute the following semantics.
			
			\begin{proposition}\label{prop:gauss_semantics}
				In $\quadrel$, we have
				\[ \sem{\underline{\mathcal N(\mu,\Sigma)}}(x) = \frac 1 2 \langle (x-\mu),\Omega(x-\mu) \rangle + \iv{(x-\mu) \in \im(\Sigma)} \]
				where $\Omega$ is any generalized inverse of $\Sigma$.
			\end{proposition}
			
			\begin{proposition}\label{prop:gaussex_semantics_nf}
				Consider two states in $\GQA$ with subspaces $D_1,D_2 \subseteq \R^n$, $\mu_1 \in D_1^\bot, \mu_2 \in D_2^\bot$, $\im(\Sigma_1) \subseteq D_1^\bot, \im(\Sigma_2) \subseteq D_2^\bot$ such that
				\begin{equation} \sem{\input{gaussex_states_1.tikz}} = \sem{\input{gaussex_states_2.tikz}} \label{eq:nf_function} \end{equation}
				Then $D_1 = D_2$, $\mu_1 = \mu_2$ and $\Sigma_1 = \Sigma_2$.
			\end{proposition}
			
			\begin{proof}
				Let $f : \R^n \to [0,\infty]$ be the quadratic relation denoted by \eqref{eq:nf_function}. Then we can read off the space $D_1 = D_2$ as the largest subspace along which $f$ is constant (\Cref{prop:projection_formula_semantics}). Restricting $f$ along the complement $D^\bot$, we obtain the denotation of the Gaussian distribution in \Cref{prop:gauss_semantics}, from which $\mu$ and $\Sigma$ can be determined.
			\end{proof}
			
			\subsection{Missing Proofs of \Cref{sec:functors}}\label[app]{app:functoriality}

			\begin{proof}[Proof of \Cref{prop:logdensity}]
				By definition, the map $f$ is equal in $\GQA$ to a string diagram of shape \eqref{eq:gauss_normal_form}. \Cref{prop:gaussex_semantics_nf} ensures any string diagrammatic representation of shape \eqref{eq:gauss_normal_form} yields the same semantics. Then its interpretation in $\quadrel$ is computed in detail in \Cref{prop:gauss_semantics}. 
			\end{proof}
			
			\begin{proof}[Proof of \Cref{prop:faithfulconservative}]
				Faithfulness is clear as every Gaussian map can be reconstructed from its log-density. For conservativity, if two $\GQAfwd$-terms $s,t$ are provably equal in $\GQA$, then $\sem{s} = \sem{t}$ in $\quadrel$. By faithfulness, $\semG{s} = \semG{t}$ in $\gauss$, hence $s=t$ is provable in $\GQAfwd$ by completeness.
			\end{proof}

			\subsection{More Details on Ordinary Least Squares in $\GQA$}\label[app]{app:ols}

\input{ols}

						\subsection{Typing Judgements for $\GPL$}\label[app]{app:typingjudg}
			
			Types $\tau$ are generated from a basic type $\rv$ denoting \emph{real} or \emph{random variable}, pair types and unit type $I$. That is, $\tau ::= \rv \s \unit \s \tau \ast \tau$. Typing judgements are as follows.
			
			\[ \infer{\Gamma, x : \tau, \Gamma' \vdash x : \tau}{}  
			\qquad
			\infer{\Gamma \vdash () : \unit}{}
			\qquad
			\infer{\Gamma \vdash (s,t) : \sigma \ast \tau}{\Gamma \vdash s : \sigma \quad \Gamma \vdash t : \tau}  
			\]
			\[ \infer{\Gamma \vdash s + t : \rv}{\Gamma \vdash s : \rv \quad \Gamma \vdash t : \rv} 
			\qquad
			\infer{\Gamma \vdash \alpha \cdot t : \rv}{\Gamma \vdash t : \rv}
			\qquad
			\infer{\Gamma \vdash \underline{\beta} : \rv}{}
			\]
			\[ \infer{\Gamma \vdash \normal() : \rv}{} 
			\qquad
			\infer{\Gamma \vdash (s \eq t) : \unit}{\Gamma \vdash s : \rv \quad \Gamma \vdash t : \rv}
			\]
			\[
			\infer{\Gamma \vdash s;t : \tau}{\Gamma \vdash s : I \quad \Gamma \vdash t : \tau}
			\]
			\[ \infer{\Gamma \vdash \letin x s t : \tau}{\Gamma \vdash s : \sigma \quad \Gamma, x : \sigma \vdash t : \tau} \qquad
			\infer{\Gamma \vdash \pi_i\,e : \tau_i}{\Gamma \vdash e : \tau_1 \ast \tau_2}
			\]
			
\end{document}

%% file: macros.tex
\usepackage{amssymb}
\usepackage{stmaryrd}
\usepackage{hyperref}
\usepackage{thmtools} 
\usepackage{thm-restate}
\usepackage[justification=centering]{caption}

\usepackage[all]{xy}

\usepackage{tikzit}
\input{local.tikzstyles}
\usetikzlibrary{external}
\tikzset{x=0.9em, y=2ex, baseline=-0.5ex}
\tikzset{ihbase/.style={inner sep=0,circle,draw,fill=lightgray,minimum size=0.4em,node contents={}}}
\tikzset{ihblack/.style={ihbase,fill=black}}
\tikzset{ihwhite/.style={ihbase,fill=white}}
\tikzset{mat/.style={draw,fill=white,rectangle,node font=\scriptsize}}
\tikzset{ha/.style={mat,rounded rectangle,rounded rectangle left arc=none}}
\tikzset{haop/.style={mat,rounded rectangle,rounded rectangle right arc=none}}
\tikzset{blackha/.style={mat,rounded rectangle,rounded rectangle left arc=none,font=\color{white},fill=black}}
\tikzset{blackhaop/.style={mat,rounded rectangle,rounded rectangle right arc=none,font=\color{white},fill=black}}
\tikzset{anti/.style={inner sep=0,isosceles triangle,fill=black,draw=black, minimum width=0.75em, node contents={}}}
\tikzset{antiop/.style={anti,shape border rotate=180}}
\tikzset{antisq/.style={inner sep=0,rectangle,fill=black, minimum height=1em, minimum width=0.6em, node contents={}}}
\tikzset{count/.style={above,inner ysep=0.15em,font=\scriptsize}}
\tikzset{axiom/.style={above,font=\small}}
\tikzset{dir/.style={-Latex}}
\tikzset{st/.style={decoration={markings,
    mark={at position 0.5 with {\draw (0, 2pt) to (0, -2pt);}}},
    postaction=decorate}}
\pgfdeclarelayer{edgelayer}
\pgfdeclarelayer{nodelayer}
\pgfsetlayers{edgelayer,nodelayer,main}
\input{local.tikzstyles}
\definecolor{light-gray}{gray}{.5}

\renewcommand{\tikzfig}[1]{
\InputIfFileExists{#1.tikz}{}{\input{./tikz/#1.tikz}}
}

\newcommand{\myeq}[1]{\mathrel{\overset{\makebox[0pt]{\mbox{\normalfont\tiny\sffamily #1}}}{=}}}

\newcommand{\diagbox}[3]{
\begin{tikzpicture}
	\begin{pgfonlayer}{nodelayer}
		\node [style=basic box] (0) at (0, 0) {$#1$};
		\node [style=none] (1) at (1.5, 0) {};
		\node [style=none] (2) at (-1.5, 0) {};
		\node [style=none] (3) at (1.5, 0.5) {\scriptsize $#3$};
		\node [style=none] (4) at (-1.5, 0.5) {\scriptsize $#2$};
	\end{pgfonlayer}
	\begin{pgfonlayer}{edgelayer}
		\draw (2.center) to (0);
		\draw (0) to (1.center);
	\end{pgfonlayer}
\end{tikzpicture}
}

\newcommand{\genericcomult}[2]{
  \begin{tikzpicture}
    \node at (1, 0) [ihbase,solid,name=copy,#1];
    \draw[#2] (copy) .. controls (1.25, 0.5) .. (2, 0.5);
    \draw[#2] (0, 0) -- (copy);
    \draw[#2] (copy) .. controls (1.25, -0.5) .. (2, -0.5);
  \end{tikzpicture}
}

\newcommand{\genericcounit}[2]{
  \tikz \draw[#2] (0, 0) -- (1, 0) node[ihbase,#1, solid];
}

\newcommand{\genericmult}[2]{
  \tikz {
    \node at (1,0) (copy) [ihbase,#1,solid];
    \draw[#2] (0,  0.5) .. controls (0.75,  0.5) .. (copy);
    \draw[#2] (0, -0.5) .. controls (0.75, -0.5) .. (copy);
    \draw[#2] (copy) -- (2, 0);
  }
}

\newcommand{\genericunit}[2]{
  \tikz \draw[#2] (0, 0) node[ihbase,#1, solid] -- (1, 0);
}

\newcommand{\Bcomult}{\genericcomult{ihblack}{}}

\newcommand{\Bcounit}{\genericcounit{ihblack}{}}

\newcommand{\Bunit}{\genericunit{ihblack}{}}

\newcommand{\NN}{\genericunit{gray,regular polygon, regular polygon sides=3, minimum width=0.3cm, shape border rotate=90, inner sep=0pt}{}}

\newcommand{\Wmult}{\genericmult{ihwhite}{}}

\newcommand{\Wunit}{\genericunit{ihwhite}{}}

\newcommand{\Wcomult}{\genericcomult{ihwhite}{}}

\newcommand{\Wcounit}{\genericcounit{ihwhite}{}}

\newcommand\scalar[1]{
	\tikz {
		\node[ha] (ha) {$#1$};
		\draw (ha.west) -- ++(-0.75, 0);
		\draw (ha.east) -- ++(0.75, 0);
	}
}

\newcommand\coscalar{\input{coscalar.tikz}}

\newcommand{\nullm}{\ensuremath{( \, )}}
\newcommand{\zeromat}{\ensuremath{[\, ]}}

\usepackage{color}
\definecolor{deepblue}{rgb}{0,0,0.5}
\definecolor{deepred}{rgb}{0.6,0,0}
\definecolor{deepgreen}{rgb}{0,0.5,0}
\definecolor{darkgray}{rgb}{0.5,0.5,0.5}

\usepackage{listings}

\DeclareFixedFont{\ttb}{T1}{txtt}{bx}{n}{9} 
\DeclareFixedFont{\ttm}{T1}{txtt}{m}{n}{9}  

\lstdefinestyle{python_ppl}{
	language=Python,
	basicstyle=\ttm,
	otherkeywords={let, to},          
	keywordstyle=\ttb\color{deepblue},
	emph={Gauss,N,condition, =:=, observe, sample, score, normal, gp_sample},     
	emphstyle=\ttb\color{deepred},    
	stringstyle=\color{deepgreen}  
}

\lstdefinelanguage{CustomML}{
	keywords={match, with, rec, true, false, fun, return, let, in, if, then, else, type, val, module, sig, end, ref, struct},
	keywordstyle=\color{deepblue}\bfseries,
	ndkeywords={ref},
	ndkeywordstyle=\color{darkgray}\bfseries,
	identifierstyle=\color{black},
	sensitive=false,
	comment=[l]{//},
	morecomment=[s]{/*}{*/},
	commentstyle=\color{darkgray}\ttfamily,
	stringstyle=\color{red}\ttfamily,
	morestring=[b]',
	morestring=[b]"
}

\lstdefinestyle{ml_ppl}{
	language=CustomML,
	basicstyle=\ttm,
	otherkeywords={},          
	keywordstyle=\ttb\color{deepblue},
	emph={Gauss,N,condition, =:=, observe, sample, score, normal, gp_sample},     
	emphstyle=\ttb\color{deepred},    
	stringstyle=\color{deepgreen}
}

\lstdefinelanguage{JavaScript}{
	keywords={typeof, new, true, false, catch, function, return, null, catch, switch, var, if, in, while, do, else, case, break},
	keywordstyle=\color{blue}\bfseries,
	ndkeywords={class, export, boolean, throw, implements, import, this},
	ndkeywordstyle=\color{darkgray}\bfseries,
	identifierstyle=\color{black},
	sensitive=false,
	comment=[l]{//},
	morecomment=[s]{/*}{*/},
	commentstyle=\color{darkgray}\ttfamily,
	stringstyle=\color{red}\ttfamily,
	morestring=[b]',
	morestring=[b]"
}

\lstdefinestyle{webppl}{
	language=JavaScript,
	basicstyle=\ttm,
	otherkeywords={let},            
	keywordstyle=\ttb\color{deepblue},
	emph={flip, condition, factor, sample, normal, score, observe},
	emphstyle=\ttb\color{deepred},    
	stringstyle=\color{deepgreen}
}

\lstdefinestyle{prolog}{
	language=Prolog,
	basicstyle=\ttm,         
	keywordstyle=\ttb\color{deepblue},
	emphstyle=\ttb\color{deepred},    
	stringstyle=\color{deepgreen}, 
}

\lstdefinestyle{funprolog}{
	language=Prolog,
	basicstyle=\ttm,
	otherkeywords={let, in, return},            
	keywordstyle=\ttb\color{deepblue},
	emphstyle=\ttb\color{deepred},    
	stringstyle=\color{deepgreen}, 
}

\newcommand*{\mlstinline}[1]{\text{\lstinline|#1|}}
\newcommand*{\code}[1]{\lstinline|#1|}

\newcommand{\R}{\mathbb R}
\newcommand{\exR}{[-\infty,\infty]}
\newcommand{\id}{\mathsf{id}}
\newcommand{\iv}[1]{\{\!|\,{#1}\,|\!\}}

\newcommand{\sem}[1]{\left\llbracket {#1} \right\rrbracket}
\newcommand{\semG}[1]{\left\langle {#1} \right\rangle}
\newcommand{\semEG}[1]{\left\langle\!\left\langle {#1} \right\rangle\!\right\rangle}

\newcommand{\inff}[2]{\inf_{{#1}} \left\{{#2}\right\}}

\newcommand{\freeprop}[1]{\mathsf{FreeP}_{\scriptscriptstyle #1}}

\newcommand{\Thlin}{\mathbb L\mathrm{in}}

\newcommand{\Thaff}{\mathbb A\mathrm{ff}}

\newcommand{\Threl}{\mathbb R\mathrm{el}}

\newcommand{\Thquad}{\mathbb Q\mathrm{uad}}

\newcommand{\Thpart}{\mathbb P\mathrm{art}}

\newcommand{\catname}[1]{\mathbf{#1}}
\newcommand{\gauss}{\catname{Gauss}}
\newcommand{\gaussex}{\catname{GaussEx}}

\newcommand{\vect}{\catname{Vect}}

\newcommand{\affrel}{\catname{AffRel}}

\newcommand{\aff}{\catname{AffVect}}
\newcommand{\quadrel}{\catname{QuadRel}}
\newcommand{\GAA}{\catname{GAA}}

\newcommand{\GQA}{\catname{GQA}}
\newcommand{\GAAfwd}{\overrightarrow{\catname{GAA}}}
\newcommand{\GQAfwd}{\overrightarrow{\catname{GQA}}}
\newcommand{\GQAfwdex}{\overrightarrow{\catname{GQA}}_{\scriptscriptstyle{E}}}

\newcommand{\mat}[1]{\underline{#1}}
\newcommand{\const}[1]{\underline{#1}}
\newcommand{\dist}[1]{\underline{#1}}
\newcommand{\subspace}[1]{\underline{#1}}

\newcommand{\orth}{\mathrm O}

\newcommand{\gl}{\mathrm{GL}}

\newcommand{\im}{\mathrm{im}}

\newcommand{\s}{\,\,|\,\,}

\newcommand{\N}{\mathcal N}
\renewcommand{\vec}[1]{\mathbf{#1}}

\newcommand{\defeq}{\stackrel{\mathrm{def.}}=}

\newcommand{\vecv}{\mathbf{v}}
\newcommand{\vecw}{\mathbf{w}}
\newcommand{\vecu}{\mathbf{u}}

\newcommand{\eq}{\mathrel{\scalebox{0.4}[1]{${=}$}{:}{\scalebox{0.4}[1]{${=}$}}}}

\usepackage{amsmath}

\usepackage[mathscr]{eucal}

\usepackage{bm}

\usepackage{etoolbox} 
\usepackage{mathtools}

\newcommand{\resistor}{\lower5pt\hbox{$\includegraphics[height=.8cm]{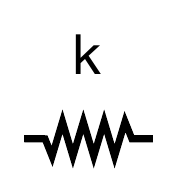}$}}
\newcommand{\vsource}{\lower5pt\hbox{$\includegraphics[height=.8cm]{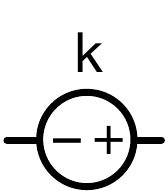}$}}
\newcommand{\csource}{\lower5pt\hbox{$\includegraphics[height=.8cm]{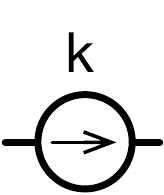}$}}
\newcommand{\nvsource}{\lower5pt\hbox{$\includegraphics[height=.8cm]{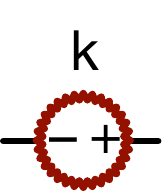}$}}
\newcommand{\ncsource}{\lower5pt\hbox{$\includegraphics[height=.8cm]{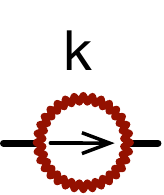}$}}
\newcommand{\inductor}{\lower5pt\hbox{$\includegraphics[height=.8cm]{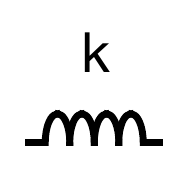}$}}
\newcommand{\capacitor}{\lower5pt\hbox{$\includegraphics[height=.8cm]{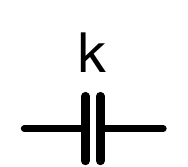}$}}
\newcommand{\junction}{\lower10pt\hbox{$\includegraphics[height=.8cm]{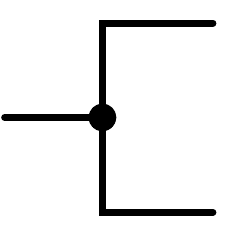}$}}

\newcommand{\cojunction}{\lower10pt\hbox{$\includegraphics[height=.8cm]{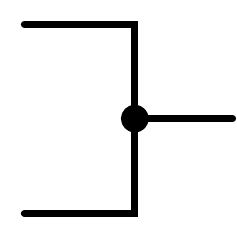}$}}

\newcommand{\GPL}{\textsc{GPL}}

\newcommand{\keyword}[1]{\mathsf{#1}}

\newcommand{\uniform}{\keyword{uniform}}

\newcommand{\vsem}[1]{\left(\!\left|{#1}\right|\!\right)}

\renewcommand{\N}{\mathcal N}

\usepackage{proof}

\newcommand{\term}[1]{\mathrm{#1}}

\newcommand{\letin}[2]{\term{let}\,{#1}\,={#2}\,\term{in}\,}
\newcommand{\pmatch}[2]{\term{match}\,{#2}\,\term{as}\,{#1}.}

\newcommand{\rv}{\mathrm{R}}

\newcommand{\unit}{\mathrm{I}}

\newcommand{\normal}{\mathrm{normal}}

%% file: local.tikzstyles
\tikzstyle{none}=[inner sep=0pt]
\tikzstyle{plain}=[inner sep=0pt]
\tikzstyle{black}=[circle, draw=black, fill=black, inner sep=0pt, minimum size=4pt]
\tikzstyle{black-faded}=[circle, draw=light-gray, fill=light-gray, inner sep=0pt, minimum size=4pt]
\tikzstyle{white}=[circle, draw=black, fill=white, inner sep=0pt, minimum size=4.5pt]
\tikzstyle{z spider}=[circle, draw=black, fill=white, inner sep=2pt, minimum size=4.5pt]
\tikzstyle{x spider}=[circle, draw=black, fill=light-gray, inner sep=2pt, minimum size=4pt]
\tikzstyle{white-faded}=[circle, draw=light-gray, fill=white, inner sep=0pt, minimum size=4.5pt]

\tikzstyle{delay}=[fill=black, regular polygon, regular polygon sides=3, rotate=-90, scale=.55]
\tikzstyle{delay-op}=[fill=black, regular polygon, regular polygon sides=3, rotate=90, scale=.55]
\tikzstyle{reg}=[draw, fill=white, rounded rectangle, rounded rectangle left arc=none, minimum height=1.2em, minimum width=1.4em, node font={\scriptsize}]
\tikzstyle{effect}=[draw, fill=white, rectangle, minimum height=1.2em, minimum width=1em]
\tikzstyle{coreg}=[draw, fill=white, rounded rectangle, rounded rectangle right arc=none, minimum height=1.2em, minimum width=1.4em, node font={\scriptsize}]

\tikzstyle{state}=[draw, fill=white, rectangle, minimum height=1.2em, minimum width=1em]
\tikzstyle{hyperedge}=[draw, fill=white, rectangle, rounded corners, minimum height=1.2em, minimum width=1em]
\tikzstyle{box}=[shape=rectangle, text height=1.5ex, text depth=0.25ex, yshift=0.2mm, fill=white, draw=black, minimum height=3mm, minimum width=5mm, font={\small}]
\tikzstyle{basic box}=[draw, fill=white, rectangle, minimum height=1.2em, minimum width=1em]
\tikzstyle{small box}=[draw, fill=white, rectangle, rounded corners, minimum height=1.2em, minimum width=1.4em, node font={\scriptsize}]
\tikzstyle{rcoreg}=[draw=red, fill=white, rounded rectangle, rounded rectangle right arc=none, minimum height=1.2em, minimum width=1.4em, node font={\scriptsize}]
\tikzstyle{regb}=[draw, fill=black, rounded rectangle, rounded rectangle left arc=none, minimum height=1.2em, minimum width=1.4em, node font={\scriptsize}]
\tikzstyle{regbw}=[draw, left color=black, right color=white, middle color=white, rounded rectangle, rounded rectangle left arc=none, minimum height=1.2em, minimum width=1.4em, node font={\scriptsize}]
\tikzstyle{regwb}=[draw, left color=white, right color=black, middle color=white, rounded rectangle, rounded rectangle left arc=none, minimum height=1.2em, minimum width=1.4em, node font={\scriptsize}]
\tikzstyle{coregb}=[draw, fill=black, rounded rectangle, rounded rectangle right arc=none, minimum height=1.2em, minimum width=1.4em, node font={\scriptsize}]
\tikzstyle{coregbw}=[draw, left color=black, right color=white, middle color=white, rounded rectangle, rounded rectangle right arc=none, minimum height=1.2em, minimum width=1.4em, node font={\scriptsize}]
\tikzstyle{coregwb}=[draw, left color=white, right color=black, middle color=white, rounded rectangle, rounded rectangle right arc=none, minimum height=1.2em, minimum width=1.4em, node font={\scriptsize}]
\tikzstyle{rn}=[circle, draw=red, fill=red, inner sep=0pt, minimum size=4pt]
\tikzstyle{wrn}=[circle, draw=red, fill=white, inner sep=0pt, minimum size=4pt]
\tikzstyle{place}=[circle, draw=black, fill=white, inner sep=0pt, minimum size=9pt]
\tikzstyle{act}=[circle, draw=black, fill=white, inner sep=0pt, minimum size=4.5pt]
\tikzstyle{coact}=[draw, fill=white, rounded rectangle, rounded rectangle right arc=none, minimum height=.7em, minimum width=.9em, node font={\scriptsize}]
\tikzstyle{basic rounded box}=[draw, fill=white, rectangle, rounded corners, minimum height=1.2em, minimum width=1.4em]
\tikzstyle{small rounded box}=[draw, fill=white, rectangle, rounded corners, minimum height=1.2em, minimum width=1.4em, node font={\scriptsize}]
\tikzstyle{red dot}=[fill=red, draw=black, shape=circle, scale=0.3]
\tikzstyle{blue dot}=[fill=blue, draw=black, shape=circle, scale=0.3]
\tikzstyle{green dot}=[fill=green, draw=black, shape=circle, scale=0.6]
\tikzstyle{medium box}=[fill=white, draw=black, shape=rectangle, minimum width=0.3cm, minimum height=0.5cm]
\tikzstyle{white square}=[fill=white, draw, inner sep=0.6mm, minimum height=1.5mm, minimum width=1.5mm, shape=rectangle, shape=rectangle]
\tikzstyle{vertex}=[inner sep=0mm, minimum size=1mm, shape=circle, draw=black, fill=black]
\tikzstyle{vertex set}=[inner sep=0mm, minimum size=1mm, shape=circle, draw=black, fill=white, font={\footnotesize\boldmath}]
\tikzstyle{s flat}=[fill=white, draw=black, shape=rectangle, minimum width=8mm, minimum height=5mm]
\tikzstyle{black dot}=[fill=black, draw=black, shape=circle, scale=0.3]
\tikzstyle{empty dot}=[fill=none, draw=black, shape=circle, scale=0.3]
\tikzstyle{l flat}=[fill=white, draw=black, shape=rectangle, minimum width=1.8cm, minimum height=0.3cm]
\tikzstyle{s rect}=[fill=white, draw=black, shape=rectangle, minimum width=0.1cm, minimum height=0.1cm]
\tikzstyle{s vert}=[fill=white, draw=black, shape=rectangle, minimum width=5mm, minimum height=8mm]
\tikzstyle{m vert}=[fill=white, draw=black, shape=rectangle, minimum width=5mm, minimum height=12mm]
\tikzstyle{m flat}=[fill=white, draw=black, shape=rectangle, minimum width=6mm, minimum height=5mm]
\tikzstyle{mm flat}=[fill=white, draw=black, shape=rectangle, minimum height=5mm, minimum width=10mm]
\tikzstyle{mmm flat}=[fill=white, draw=black, shape=rectangle, minimum height=5mm, minimum width=12mm]
\tikzstyle{grey dot}=[fill={rgb,255: red,191; green,191; blue,191}, draw={rgb,255: red,191; green,191; blue,191}, shape=circle, scale=0.3]
\tikzstyle{mm vert}=[fill=white, draw=black, shape=rectangle, minimum width=5mm, minimum height=14mm]
\tikzstyle{20mm vert}=[fill=white, draw=black, shape=rectangle, minimum width=5mm, minimum height=20mm]
\tikzstyle{16mm vert}=[fill=white, draw=black, shape=rectangle, minimum width=5mm, minimum height=16mm]
\tikzstyle{18mm vert}=[fill=white, draw=black, shape=rectangle, minimum width=5mm, minimum height=18mm]
\tikzstyle{morphism}=[fill=white, draw=black, shape=rectangle]
\tikzstyle{medium box}=[fill=white, draw=black, shape=rectangle, minimum width=0.8cm, minimum height=0.9cm]
\tikzstyle{large morphism}=[fill=white, draw=black, shape=rectangle, minimum width=1cm, minimum height=2cm]
\tikzstyle{bn}=[fill=black, draw=black, shape=circle, inner sep=1.5pt]
\tikzstyle{wn}=[fill=white, draw=black, shape=circle, inner sep=1.5pt]
\tikzstyle{nn}=[fill=gray, draw=gray, regular polygon, regular polygon sides=3, minimum width=0.3cm, shape border rotate=90, inner sep=0pt]
\tikzstyle{nneffect}=[fill=gray, draw=gray, regular polygon, regular polygon sides=3, minimum width=0.3cm, shape border rotate=-90, inner sep=0pt]
\tikzstyle{state}=[fill=white, draw=black, regular polygon, regular polygon sides=3, minimum width=0.3cm, shape border rotate=270, inner sep=0.5pt, text width=2pt]
\tikzstyle{effect}=[fill=white, draw=black, regular polygon, regular polygon sides=3, minimum width=0.8cm, shape border rotate=0, inner sep=0pt]
\tikzstyle{medium state}=[fill=white, draw=black, regular polygon, regular polygon sides=3, minimum width=1.3cm, inner sep=0pt, shape border rotate=180]
\tikzstyle{medium effect}=[fill=white, draw=black, regular polygon, regular polygon sides=3, minimum width=1.3cm, inner sep=0pt, shape border rotate=0]
\tikzstyle{large state}=[fill=white, draw=black, regular polygon, regular polygon sides=3, minimum width=2.2cm, shape border rotate=180, inner sep=0pt]
\tikzstyle{rotated state}=[fill=white, draw=black, regular polygon, regular polygon sides=3, minimum width=0.3cm, shape border rotate=90, inner sep=0.5pt, text width=2pt]

\tikzstyle{dashes}=[-, dashed]
\tikzstyle{dashed box}=[-, dashed]
\tikzstyle{thick}=[-, draw=black, line width=0.4mm]
\tikzstyle{right arrow}=[->]
\tikzstyle{left arrow}=[<-]
\tikzstyle{grey fill}=[-, fill={rgb,255: red,191; green,191; blue,191}, draw={rgb,255: red,191; green,191; blue,191}]
\tikzstyle{blue fill}=[-, fill=cyan, draw=cyan]
\tikzstyle{yellow fill}=[-, fill=yellow, draw=yellow]
\tikzstyle{green fill}=[-, fill=green, draw=green]
\tikzstyle{red wire}=[-, draw=red]
\tikzstyle{blue wire}=[-, draw=blue]
\tikzstyle{grey wire}=[-, draw={rgb,255: red,128; green,128; blue,128}]
\tikzstyle{green wire}=[-, draw=green]
\tikzstyle{yellow wire}=[-, draw=yellow]
\tikzstyle{black fill}=[-, fill=black]
\tikzstyle{white thick}=[-, draw=white, line width=4.5pt]

%% file: tikz/coscalar.tikz
\begin{tikzpicture}
	\begin{pgfonlayer}{nodelayer}
		\node [style=none] (0) at (-1.75, 0) {};
		\node [style=coreg] (2) at (0, 0) {$k$};
		\node [style=none] (3) at (1.75, 0) {};
	\end{pgfonlayer}
	\begin{pgfonlayer}{edgelayer}
		\draw (3.center) to (2);
		\draw (0.center) to (2);
	\end{pgfonlayer}
\end{tikzpicture}

%% file: tikz/bcounit.tikz
\begin{tikzpicture}
	\begin{pgfonlayer}{nodelayer}
		\node [style=black] (1) at (1.75, 0) {};
		\node [style=none] (4) at (0.5, 0) {};
	\end{pgfonlayer}
	\begin{pgfonlayer}{edgelayer}
		\draw (1) to (4.center);
	\end{pgfonlayer}
\end{tikzpicture}

%% file: tikz/wunit.tikz
\begin{tikzpicture}
	\begin{pgfonlayer}{nodelayer}
		\node [style=none] (8) at (0.75, 0) {};
		\node [style=none] (17) at (0.25, 0) {};
		\node [style=white] (23) at (-0.5, 0) {};
	\end{pgfonlayer}
	\begin{pgfonlayer}{edgelayer}
		\draw (17.center) to (8.center);
		\draw (23) to (17.center);
	\end{pgfonlayer}
\end{tikzpicture}

%% file: tikz/foot.tikz
\begin{tikzpicture}
	\begin{pgfonlayer}{nodelayer}
		\node [style=none] (2) at (-1, 0) {};
		\node [style=none] (3) at (0.75, 0) {};
		\node [style=none] (4) at (-1, 0.5) {};
		\node [style=none] (5) at (-1, -0.5) {};
	\end{pgfonlayer}
	\begin{pgfonlayer}{edgelayer}
		\draw (2.center) to (3.center);
		\draw (5.center) to (4.center);
	\end{pgfonlayer}
\end{tikzpicture}

%% file: tikz/bunit.tikz
\begin{tikzpicture}
	\begin{pgfonlayer}{nodelayer}
		\node [style=black] (1) at (0.5, 0) {};
		\node [style=none] (4) at (1.75, 0) {};
	\end{pgfonlayer}
	\begin{pgfonlayer}{edgelayer}
		\draw (1) to (4.center);
	\end{pgfonlayer}
\end{tikzpicture}

%% file: tikz/wcounit.tikz
\begin{tikzpicture}
	\begin{pgfonlayer}{nodelayer}
		\node [style=white] (1) at (1.75, 0) {};
		\node [style=none] (4) at (0.5, 0) {};
	\end{pgfonlayer}
	\begin{pgfonlayer}{edgelayer}
		\draw (1) to (4.center);
	\end{pgfonlayer}
\end{tikzpicture}

%% file: tikz/scalar.tikz
\begin{tikzpicture}
	\begin{pgfonlayer}{nodelayer}
		\node [style=none] (0) at (-1.5, 0) {};
		\node [style=none] (1) at (1.75, 0) {};
		\node [style=reg] (2) at (0, 0) {$\scriptstyle k$};
	\end{pgfonlayer}
	\begin{pgfonlayer}{edgelayer}
		\draw (0.center) to (2);
		\draw (2) to (1.center);
	\end{pgfonlayer}
\end{tikzpicture}

%% file: tikz/add_example_2Right.tikz
\begin{tikzpicture}[scale=1]
	\begin{pgfonlayer}{nodelayer}
		\node [style=nn] (4) at (-2, 0) {};
		\node [style=reg] (6) at (0, 0) {$\sqrt 2$};
		\node [style=none] (7) at (2, 0) {};
	\end{pgfonlayer}
	\begin{pgfonlayer}{edgelayer}
		\draw (4) to (6);
		\draw (6) to (7.center);
	\end{pgfonlayer}
\end{tikzpicture}

%% file: tikz/smc/id.tikz
\begin{tikzpicture}
	\begin{pgfonlayer}{nodelayer}
		\node [style=none] (71) at (1.25, 0) {};
		\node [style=none] (72) at (-1.75, 0) {};
	\end{pgfonlayer}
	\begin{pgfonlayer}{edgelayer}
		\draw (72.center) to (71.center);
	\end{pgfonlayer}
\end{tikzpicture}

%% file: tikz/cofoot.tikz
\begin{tikzpicture}
	\begin{pgfonlayer}{nodelayer}
		\node [style=none] (2) at (0.75, 0) {};
		\node [style=none] (3) at (-1, 0) {};
		\node [style=none] (4) at (0.75, 0.5) {};
		\node [style=none] (5) at (0.75, -0.5) {};
	\end{pgfonlayer}
	\begin{pgfonlayer}{edgelayer}
		\draw (2.center) to (3.center);
		\draw (5.center) to (4.center);
	\end{pgfonlayer}
\end{tikzpicture}

%% file: tikz/false.tikz
\begin{tikzpicture}
	\begin{pgfonlayer}{nodelayer}
		\node [style=none] (0) at (-0.5, 0) {};
		\node [style=white] (1) at (0.5, 0) {};
		\node [style=none] (2) at (-0.5, -0.25) {};
		\node [style=none] (3) at (-0.5, 0.25) {};
	\end{pgfonlayer}
	\begin{pgfonlayer}{edgelayer}
		\draw (0.center) to (1);
		\draw (3.center) to (2.center);
	\end{pgfonlayer}
\end{tikzpicture}

%% file: tikz/cdiagram-thick.tikz
\begin{tikzpicture}
	\begin{pgfonlayer}{nodelayer}
		\node [style=basic box] (0) at (0, 0) {$c$};
		\node [style=none] (1) at (-1.75, 0) {};
		\node [style=none] (2) at (1.5, 0) {};
	\end{pgfonlayer}
	\begin{pgfonlayer}{edgelayer}
		\draw [style=thick] (1.center) to (0);
		\draw [style=thick] (0) to (2.center);
	\end{pgfonlayer}
\end{tikzpicture}

%% file: tikz/wunit-thick.tikz
\begin{tikzpicture}
	\begin{pgfonlayer}{nodelayer}
		\node [style=none] (0) at (-1, 0) {};
		\node [style=wn] (1) at (1, 0) {};
	\end{pgfonlayer}
	\begin{pgfonlayer}{edgelayer}
		\draw [style=thick] (1) to (0.center);
	\end{pgfonlayer}
\end{tikzpicture}

%% file: tikz/matrixdiagram.tikz
\begin{tikzpicture}
	\begin{pgfonlayer}{nodelayer}
		\node [style=basic box] (0) at (0, 0) {$\underline{A}$};
		\node [style=none] (1) at (-1.75, 0) {};
		\node [style=none] (2) at (1.5, 0) {};
	\end{pgfonlayer}
	\begin{pgfonlayer}{edgelayer}
		\draw [style=thick] (1.center) to (0);
		\draw [style=thick] (0) to (2.center);
	\end{pgfonlayer}
\end{tikzpicture}

%% file: tikz/subspacediagram.tikz
\begin{tikzpicture}
	\begin{pgfonlayer}{nodelayer}
		\node [style=basic box] (0) at (0, 0) {$\underline{A}$};
		\node [style=black] (1) at (-1.75, 0) {};
		\node [style=none] (2) at (1.5, 0) {};
	\end{pgfonlayer}
	\begin{pgfonlayer}{edgelayer}
		\draw [style=thick] (1) to (0);
		\draw [style=thick] (0) to (2.center);
	\end{pgfonlayer}
\end{tikzpicture}

%% file: tikz/subspacediagram_compressed.tikz
\begin{tikzpicture}
	\begin{pgfonlayer}{nodelayer}
		\node [style=morphism] (0) at (0, 0) {$\subspace{S}$};
		\node [style=none] (2) at (2.5, 0) {};
	\end{pgfonlayer}
	\begin{pgfonlayer}{edgelayer}
		\draw [style=thick] (0) to (2.center);
	\end{pgfonlayer}
\end{tikzpicture}

%% file: tikz/scalar-factor-notation.tikz
\begin{tikzpicture}[scale=1]
	\begin{pgfonlayer}{nodelayer}
		\node [style=none] (44) at (1.25, 0) {};
		\node [style=box] (46) at (-1, 0) {$\const{c}$};
		\node [style=none] (47) at (-0.75, 0) {};
	\end{pgfonlayer}
	\begin{pgfonlayer}{edgelayer}
		\draw (47.center) to (44.center);
	\end{pgfonlayer}
\end{tikzpicture}

%% file: tikz/doublecoNN.tikz
\begin{tikzpicture}
	\begin{pgfonlayer}{nodelayer}
		\node [style=nneffect] (0) at (1, 0.75) {};
		\node [style=nneffect] (1) at (1, -0.5) {};
		\node [style=none] (2) at (-1, 0.75) {};
		\node [style=none] (3) at (-1, -0.5) {};
	\end{pgfonlayer}
	\begin{pgfonlayer}{edgelayer}
		\draw (0) to (2.center);
		\draw (1) to (3.center);
	\end{pgfonlayer}
\end{tikzpicture}

%% file: tikz/leastsquares_exampleRHS.tikz
\begin{tikzpicture}[scale=1]
	\begin{pgfonlayer}{nodelayer}
		\node [style=coreg] (7) at (0, 0) {$\sqrt 2$};
		\node [style=nneffect] (8) at (2, 0) {};
		\node [style=none] (9) at (-1.75, 0) {};
	\end{pgfonlayer}
	\begin{pgfonlayer}{edgelayer}
		\draw (7) to (8);
		\draw (9.center) to (7);
	\end{pgfonlayer}
\end{tikzpicture}

%% file: tikz/ax/bone-black.tikz
\begin{tikzpicture}
	\begin{pgfonlayer}{nodelayer}
		\node [style=black] (0) at (0.5, -0) {};
		\node [style=black] (1) at (-0.75, -0) {};
	\end{pgfonlayer}
	\begin{pgfonlayer}{edgelayer}
		\draw (0) to (1);
	\end{pgfonlayer}
\end{tikzpicture}

%% file: tikz/ax/TI-right.tikz
\begin{tikzpicture}
	\begin{pgfonlayer}{nodelayer}
		\node [style=black] (9) at (0, 0) {};
		\node [style=none] (11) at (1.25, 0) {};
		\node [style=black] (12) at (-1.25, 0) {};
		\node [style=none] (13) at (-2.5, 0) {};
	\end{pgfonlayer}
	\begin{pgfonlayer}{edgelayer}
		\draw (9) to (11.center);
		\draw (12) to (13.center);
	\end{pgfonlayer}
\end{tikzpicture}

%% file: tikz/ax/bunitscalar.tikz
\begin{tikzpicture}
	\begin{pgfonlayer}{nodelayer}
		\node [style=reg] (0) at (0, 0) {$\scriptstyle k$};
		\node [style=black] (1) at (-1.75, 0) {};
		\node [style=none] (2) at (1.5, 0) {};
	\end{pgfonlayer}
	\begin{pgfonlayer}{edgelayer}
		\draw (1) to (0);
		\draw (0) to (2.center);
	\end{pgfonlayer}
\end{tikzpicture}

%% file: tikz/smallscalar.tikz
\begin{tikzpicture}
	\begin{pgfonlayer}{nodelayer}
		\node [style=morphism] (0) at (-1, 0) {$\underline{c}$};
		\node [style=nneffect] (1) at (1, 0) {};
	\end{pgfonlayer}
	\begin{pgfonlayer}{edgelayer}
		\draw (0) to (1);
	\end{pgfonlayer}
\end{tikzpicture}

%% file: tikz/infty.tikz
\begin{tikzpicture}
	\begin{pgfonlayer}{nodelayer}
		\node [style=none] (0) at (-1, 0) {};
		\node [style=wn] (1) at (1, 0) {};
		\node [style=none] (2) at (-1, 0.5) {};
		\node [style=none] (3) at (-1, -0.5) {};
	\end{pgfonlayer}
	\begin{pgfonlayer}{edgelayer}
		\draw (0.center) to (1);
		\draw (2.center) to (3.center);
	\end{pgfonlayer}
\end{tikzpicture}

%% file: tikz/smc/id2.tikz
\begin{tikzpicture}
	\begin{pgfonlayer}{nodelayer}
		\node [style=none] (0) at (2, -0.75) {};
		\node [style=none] (1) at (-2, -0.75) {};
		\node [style=none] (2) at (-2, 0.5) {};
		\node [style=none] (3) at (2, 0.5) {};
	\end{pgfonlayer}
	\begin{pgfonlayer}{edgelayer}
		\draw (0.center) to (1.center);
		\draw (3.center) to (2.center);
	\end{pgfonlayer}
\end{tikzpicture}

%% file: tikz/subspace_state_A.tikz
\begin{tikzpicture}[scale=1]
	\begin{pgfonlayer}{nodelayer}
		\node [style=bn] (18) at (-2, 0) {};
		\node [style=box] (21) at (0, 0) {$\mat A$};
		\node [style=none] (22) at (2, 0) {};
	\end{pgfonlayer}
	\begin{pgfonlayer}{edgelayer}
		\draw [style=thick] (21) to (18);
		\draw [style=thick] (22.center) to (21);
	\end{pgfonlayer}
\end{tikzpicture}

%% file: tikz/subspace_state_B.tikz
\begin{tikzpicture}[scale=1]
	\begin{pgfonlayer}{nodelayer}
		\node [style=bn] (18) at (-2, 0) {};
		\node [style=box] (21) at (0, 0) {$\mat B$};
		\node [style=none] (22) at (2, 0) {};
	\end{pgfonlayer}
	\begin{pgfonlayer}{edgelayer}
		\draw [style=thick] (21) to (18);
		\draw [style=thick] (22.center) to (21);
	\end{pgfonlayer}
\end{tikzpicture}

%% file: tikz/gauss_soundness_1.tikz
\begin{tikzpicture}[scale=1]
	\begin{pgfonlayer}{nodelayer}
		\node [style=nn] (18) at (-2, 0) {};
		\node [style=box] (21) at (0, 0) {$\mat R$};
		\node [style=none] (22) at (2, 0) {};
	\end{pgfonlayer}
	\begin{pgfonlayer}{edgelayer}
		\draw [style=thick] (21) to (18);
		\draw [style=thick] (22.center) to (21);
	\end{pgfonlayer}
\end{tikzpicture}

%% file: tikz/gauss_soundness_2.tikz
\begin{tikzpicture}[scale=1]
	\begin{pgfonlayer}{nodelayer}
		\node [style=nn] (18) at (-1, 0) {};
		\node [style=none] (19) at (1, 0) {};
	\end{pgfonlayer}
	\begin{pgfonlayer}{edgelayer}
		\draw [style=thick] (19.center) to (18);
	\end{pgfonlayer}
\end{tikzpicture}

%% file: tikz/oi.tikz
\begin{tikzpicture}[scale=1]
	\begin{pgfonlayer}{nodelayer}
		\node [style=nn] (18) at (-6.5, 0) {};
		\node [style=box] (21) at (-4, 0) {$\mat R$};
		\node [style=none] (22) at (-2, 0) {};
	\end{pgfonlayer}
	\begin{pgfonlayer}{edgelayer}
		\draw [style=thick] (21) to (18);
		\draw [style=thick] (22.center) to (21);
	\end{pgfonlayer}
\end{tikzpicture}

%% file: tikz/NN-thick.tikz
\begin{tikzpicture}
	\begin{pgfonlayer}{nodelayer}
		\node [style=nn] (0) at (-1, 0) {};
		\node [style=none] (1) at (1, 0) {};
	\end{pgfonlayer}
	\begin{pgfonlayer}{edgelayer}
		\draw [style=thick] (1.center) to (0);
	\end{pgfonlayer}
\end{tikzpicture}

%% file: tikz/gauss_state_A.tikz
\begin{tikzpicture}[scale=1]
	\begin{pgfonlayer}{nodelayer}
		\node [style=nn] (18) at (-2, 0) {};
		\node [style=box] (21) at (0, 0) {$\mat A$};
		\node [style=none] (22) at (2, 0) {};
	\end{pgfonlayer}
	\begin{pgfonlayer}{edgelayer}
		\draw [style=thick] (21) to (18);
		\draw [style=thick] (22.center) to (21);
	\end{pgfonlayer}
\end{tikzpicture}

%% file: tikz/gauss_state_B.tikz
\begin{tikzpicture}[scale=1]
	\begin{pgfonlayer}{nodelayer}
		\node [style=nn] (18) at (-2, 0) {};
		\node [style=box] (21) at (0, 0) {$\mat B$};
		\node [style=none] (22) at (2, 0) {};
	\end{pgfonlayer}
	\begin{pgfonlayer}{edgelayer}
		\draw [style=thick] (21) to (18);
		\draw [style=thick] (22.center) to (21);
	\end{pgfonlayer}
\end{tikzpicture}

%% file: tikz/generic-gen.tikz
\begin{tikzpicture}
	\begin{pgfonlayer}{nodelayer}
		\node [style=none] (11) at (1.25, 0) {};
		\node [style=none] (12) at (1.25, 0) {};
		\node [style=none] (13) at (-2.5, 0) {};
		\node [style=box] (14) at (-0.5, 0) {$\scriptstyle g$};
	\end{pgfonlayer}
	\begin{pgfonlayer}{edgelayer}
		\draw (12.center) to (13.center);
	\end{pgfonlayer}
\end{tikzpicture}

%% file: tikz/blackcup-nonode.tikz
\begin{tikzpicture}
	\begin{pgfonlayer}{nodelayer}
		\node [style=none] (0) at (-1, 1.25) {};
		\node [style=none] (1) at (-1, -0.75) {};
		\node [style=black] (2) at (0.5, 0.25) {};
	\end{pgfonlayer}
	\begin{pgfonlayer}{edgelayer}
		\draw [in=90, out=0] (0.center) to (2);
		\draw [in=-90, out=0] (1.center) to (2);
	\end{pgfonlayer}
\end{tikzpicture}

%% file: tikz/whitecup-nonode.tikz
\begin{tikzpicture}
	\begin{pgfonlayer}{nodelayer}
		\node [style=none] (0) at (-1, 1.25) {};
		\node [style=none] (1) at (-1, -0.75) {};
		\node [style=white] (2) at (0.5, 0.25) {};
	\end{pgfonlayer}
	\begin{pgfonlayer}{edgelayer}
		\draw [in=90, out=0] (0.center) to (2);
		\draw [in=-90, out=0] (1.center) to (2);
	\end{pgfonlayer}
\end{tikzpicture}

%% file: tikz/interp_gauss.tikz
\begin{tikzpicture}[scale=1]
	\begin{pgfonlayer}{nodelayer}
		\node [style=nn] (0) at (-2, 0) {};
		\node [style=box] (1) at (0, 0) {$\mat{A}$};
		\node [style=none] (2) at (2, 0) {};
	\end{pgfonlayer}
	\begin{pgfonlayer}{edgelayer}
		\draw [style=thick] (0) to (1);
		\draw (1) to (2.center);
	\end{pgfonlayer}
\end{tikzpicture}

%% file: ols.tex
In this appendix we expand Section~\ref{sec:ols}, showing how to apply our theory of $\GQA$ to the method of ordinary least squares in linear regression. The aim of linear regression is simple: to find a linear model that best fits a set of observations. In its usual vectorial formulation, all available observations of the regressors form the columns of a single matrix $A$ and all observations of the dependent variable form a single vector $y$; then, a linear model with parameters $x$ is expressed concisely as the system $Ax = y$. Typically, for consistency, we also assume that the regressors are linearly independent, \emph{i.e.}, that $A$ is injective. If $A$ is not invertible--as it usually isn't--this system does not admit an exact solution. We can nevertheless look for parameters $x$ such that $A{x}$ best approximates the observed values $y$. Here, `best' is interpreted in such a way that the sum of squared errors, $||y-Ax||^2$, is minimised. This function can be translated directly into the following diagram: 
$$
\sem{\input{lq-problem.tikz}} \;= \; {x \choose y}\mapsto \frac{1}{2}||y-Ax||^2
$$
The formula for the optimal $\hat x$ is well-known to be the familiar\ ordinary least squares (OLS) estimator $\hat x = (A^T A)^{-1}A^T y = A^+y$. Our aim here is to re-derive this expression starting from the diagram above, using the axioms of $\GQA$.
We then need a lemma which, in semantic terms, states that $||Ax+y||^2 = ||Ax||^2+||y||^2$ when $A^Ty=0$, \emph{i.e.}, when $y$ is in the kernel of the transpose of $A$, \emph{i.e.}, is orthogonal to the image of $A$. It is easy to see why this holds: $||Ax+y||^2 = ||Ax||^2+ 2\langle Ax,y\rangle + ||y||^2$ and since $A^Ty=0$, we have $\langle Ax,y\rangle = \langle x,A^Ty\rangle = 0$. We now derive it diagrammatically. Note that this is the only part of our derivation of the OLS estimator which uses the \emph{quadratic} fragment of GQA; the rest will be plain GAA. 
\begin{lemma}
\label{lem:ortho}
The following equality is derivable in $\GQA$: 
$$\input{ols-lemma.tikz}\;=\; \input{ols-lemma-fin.tikz}$$
\end{lemma}
\begin{proof}
We will need a few useful facts from (graphical) linear algebra. First, recall that a matrix $A$ can be written $UDV^T$, where $U,V$ are orthogonal matrices and $D$ is diagonal---this is known as the \emph{singular value decomposition} (SVD) of $A$.  Diagrammatically, this translates to the two equalities on the first line below (and the two on the second line for the transpose of $A$):
\begin{align*}\label{eq:A-svd}
\input{A.tikz} = \input{A-svd.tikz} = \input{A-svd-1.tikz}
\\
\input{A-t.tikz} = \input{transpose-svd.tikz} = \input{transpose-svd-1.tikz}
\end{align*}
where $C$ is the non-zero sub-matrix of $D$ and the remaining zeros are represented by $\Wunit$.
Going back to the initial least-square problem, we have
\begin{align*}
\input{ols-lemma.tikz}\;\myeq{SVD}\; \input{ols-lemma-1.tikz} \;\myeq{($*$)}\; \input{ols-lemma-2.tikz}
\end{align*}
where the second equality ($*$) holds, since $V$ is orthogonal and $C$ is diagonal, and so their kernel is trivial.
Then, since $U^T$ is orthogonal, we can use orthogonal invariance to obtain
\begin{align*}
&\input{ols-lemma-2.tikz}\;\;\;\myeq{Prop. \ref{prop:orthogonalinv}}\; \input{ols-lemma-3.tikz}
\\
\;&\myeq{add}\; \input{ols-lemma-4.tikz}
\; \myeq{copy}\; \input{ols-lemma-5.tikz}
\end{align*}
where the last two equalities (labelled \textsf{add} and \textsf{copy}) use the fact that matrices distribute over $\Wmult$ and $\Bcomult$ (standard facts in Graphical Affine Algebra). Then,
\begin{align*}
&\input{ols-lemma-5.tikz} \;\myeq{coun}\;\input{ols-lemma-6.tikz} \;\myeq{$\circ\bullet$-biun}\; \input{ols-lemma-7.tikz}\\
&\;\myeq{un}\;\input{ols-lemma-8.tikz}
 \;=\; \input{ols-lemma-9.tikz}\\
 &\;\myeq{un}\;\input{ols-lemma-10.tikz} \;\myeq{Z}\;\input{ols-lemma-11.tikz}\;\myeq{RI}\;\input{ols-lemma-13.tikz} 
\end{align*}
Where the last step uses again the fact that $U$ is orthogonal to apply rotational invariance. We now repeat the same steps in the other direction:
\begin{align*}
&\input{ols-lemma-13.tikz}  \;\;\myeq{$\circ\bullet$-biun}\; \input{ols-lemma-14.tikz}
\; \myeq{coun}\; \input{ols-lemma-15.tikz}
\;\myeq{copy}\; \input{ols-lemma-16.tikz}
\\
&\;=\; \input{ols-lemma-17.tikz}\;\myeq{($*$)}\; \input{ols-lemma-18.tikz} 
\;\myeq{SVD}\; \input{ols-lemma-fin.tikz}
\end{align*}
where once again, the equality ($*$) uses the fact that $V$ is orthogonal and $C$ diagonal and thus their kernel is trivial.
\end{proof}

We also record the following lemma, which is just linear algebra, so we appeal to completeness of $\GQA$ for the proof. Intuitively, it translates the fact that $I - AA^+$ is the projection onto the orthogonal space to the range of $A$, \emph{i.e.} onto the the kernel of $A^T$. (Note also that $AA^+$ is the projection onto the range of $A$.)
\begin{lemma}
\label{lem:project}
$\input{1-minus-A-A-plus.tikz}\;\;=\;\; \input{range-A-projector.tikz}$
\end{lemma}
We can now derive the OLS estimator. First, because $\Wunit$ is the unit of $\Wmult$ and $\Bcounit$ is the co-unit of $\Bcomult$, we have
\begin{align*}
\scalebox{0.8}{\input{lq-problem.tikz}} \; \myeq{coun-un} \;  \scalebox{0.8}{\input{ols-1.tikz}} \; \myeq{del} \; \scalebox{0.8}{\input{ols-2.tikz}}
\end{align*}
where the second equality (\textsf{del}) comes from the fact that any matrix composed with $\Bcounit$ is equal to $\Bcounit$.
Then, since $1+(-1)=0$, we get
\begin{align*}
&\scalebox{0.8}{\input{ols-2.tikz}} \;\myeq{0;plus;1}\; \scalebox{0.8}{\input{ols-3.tikz}}
\;\myeq{copy}\;\scalebox{0.8}{\input{ols-4.tikz}}
\\
&\;\myeq{coas}\;\scalebox{0.8}{\input{ols-5.tikz}} \;\myeq{as}\;\scalebox{0.8}{\input{ols-6.tikz}}
\;\myeq{Lemma~\ref{lem:project}}\;\;\;\scalebox{0.8}{\input{ols-7.tikz}}
\\
&\;\myeq{as}\;\scalebox{0.8}{\input{ols-8.tikz}}
\; = \;\scalebox{0.8}{\input{ols-9.tikz}}\; =\; \scalebox{0.8}{\input{ols-10.tikz}}
 \\
&\;\;\;
 \myeq{Lemma~\ref{lem:ortho}}\quad \scalebox{0.8}{\input{ols-11.tikz}}
 \;\myeq{Lemma~\ref{lem:project}}\;\;\;\scalebox{0.8}{\input{ols-fin.tikz}}
\end{align*}
The semantics of the last diagram is the quadratic function 
$${x \choose y}\mapsto \frac{1}{2}\left(||AA^+y-Ax||^2 + ||y-AA^+y||^2\right).$$
Its infimum is clearly reached at ${\hat x}=A^+y$ as we wanted: in this case $||AA^+y-A{\hat x}||^2=0$ and the remaining term $||y-AA^+y||^2 = ||y-A{\hat x}||^2$ indicates the distance between $A{\hat x}$ and $y$, \emph{i.e.}, how far we are from having successfully inverted~$A$.

%% file: tikz/A.tikz
\begin{tikzpicture}[scale=1]
	\begin{pgfonlayer}{nodelayer}
		\node [style=box] (62) at (0, 0) {$\mat{A}$};
		\node [style=none] (63) at (1.75, 0) {};
		\node [style=none] (69) at (-2, 0) {};
	\end{pgfonlayer}
	\begin{pgfonlayer}{edgelayer}
		\draw [style=thick] (62) to (63.center);
		\draw [style=thick] (69.center) to (62);
	\end{pgfonlayer}
\end{tikzpicture}

%% file: tikz/A-t.tikz
\begin{tikzpicture}[scale=1]
	\begin{pgfonlayer}{nodelayer}
		\node [style=box] (62) at (0, 0) {$\mat{A}^T$};
		\node [style=none] (63) at (1.75, 0) {};
		\node [style=none] (69) at (-2, 0) {};
	\end{pgfonlayer}
	\begin{pgfonlayer}{edgelayer}
		\draw [style=thick] (62) to (63.center);
		\draw [style=thick] (69.center) to (62);
	\end{pgfonlayer}
\end{tikzpicture}